\newif \ifproofs
\proofstrue

\documentclass[11pt]{article}
\usepackage{amssymb}
\usepackage{mathrsfs,amsmath}
\usepackage{graphicx}
\usepackage{amsmath}
\usepackage{amsthm}
\usepackage{bbm}
\usepackage{dsfont}
\usepackage{listing}
\usepackage{hyperref}
\usepackage{color}
\usepackage[ruled, vlined,linesnumbered, noend]{algorithm2e}
\usepackage{caption}
\usepackage{subcaption}
\usepackage{comment}
\allowdisplaybreaks

\newtheorem{theorem}{Theorem}

\newtheorem{lemma}{Lemma}

\newtheorem{assumption}{Assumption}
\theoremstyle{definition}
\newtheorem{definition}{Definition}
\newtheorem{remark}{Remark}

\newcommand{\abs}[1]{\left| #1 \right|}
\newcommand{\bigpar}[1]{\left( #1 \right)}
\newcommand{\bigbra}[1]{\left[ #1 \right]}
\newcommand{\bigbrace}[1]{\left\{ #1 \right\}}

\newcommand{\bigfloor}[1]{\left\lfloor #1 \right\rfloor}
\newcommand{\casewise}[1]{\left\{ #1 \right.}
\newcommand{\evaluate}[1]{\left. #1 \right|} 
\newcommand{\norm}[1]{\left| \left| #1 \right| \right|}

\newcommand{\iid}{\overset{\text{i.i.d}}{\sim}}

\newcommand{\cD}{\mathcal{D}}

\newcommand{\cL}{\mathcal{L}}

\newcommand{\cP}{\mathcal{P}}
\newcommand{\cQ}{\mathcal{Q}}

\newcommand{\cS}{\mathcal{S}}
\newcommand{\cT}{\mathcal{T}}

\newcommand{\cX}{\mathcal{X}}

\newcommand{\cZ}{\mathcal{Z}}

\usepackage{fancyhdr}
\usepackage{calc}
\fancyheadoffset[RE]{\marginparsep+\marginparwidth}

\usepackage[top=1.5in, bottom=1in, left=1in, right=1in]{geometry}

\graphicspath{ {peripherals/fig/} }

\title{Differentially Private Stochastic Convex Optimization for Network Routing Applications}

\author{
\begin{tabular}{cc}
    \begin{tabular}{c}
        Matthew Tsao \\
        Lyft Inc. \\
        \texttt{mtsao@lyft.com} 
    \end{tabular}
     & 
    \begin{tabular}{c}
        Karthik Gopalakrishnan \\
        Stanford University \\
        \texttt{gkarthik@stanford.edu} 
    \end{tabular}
    \\
    & \\
    \begin{tabular}{c}
        Kaidi Yang \\
        National University of Singapore \\
        \texttt{kaidi.yang@nus.edu.sg} 
    \end{tabular}
     & 
    \begin{tabular}{c}
        Marco Pavone \\
        Stanford University \\
        \texttt{pavone@stanford.edu} 
    \end{tabular}    
\end{tabular}
}

\begin{document}
\maketitle

\begin{abstract}

Network routing problems are common across many engineering applications. Computing optimal routing policies requires knowledge about network demand, i.e., the origin and destination (OD) of all requests in the network. However, privacy considerations make it challenging to share individual OD data that would be required for computing optimal policies. Privacy can be particularly challenging in standard network routing problems because sources and sinks can be easily identified from flow conservation constraints, making feasibility and privacy mutually exclusive. 

In this paper, we present a differentially private algorithm for network routing problems. The main ingredient is a reformulation of network routing which moves all user data-dependent parameters out of the constraint set and into the objective function. We then present an algorithm for solving this formulation based on a differentially private variant of stochastic gradient descent. In this algorithm, differential privacy is achieved by injecting noise, and one may wonder if this noise injection compromises solution quality. We prove that our algorithm is both differentially private and asymptotically optimal as the size of the training set goes to infinity. We corroborate the theoretical results with numerical experiments on a road traffic network which show that our algorithm provides differentially private and near-optimal solutions in practice.
\end{abstract}

    
\section{Introduction}

Network routing problems appear in many important topics in engineering, including traffic routing in transportation systems, power routing in electrical grids, and packet routing in distributed computer systems. Network routing problems study settings where resources must be delivered to customers through a network with limited bandwidth. The goal is typically to route resources to their respective customers as efficiently as possible, or equivalently, with as little network congestion as possible.

One key challenge in network routing problems is that the requests (i.e., network demand) are often not known in advance. Indeed, it is difficult to know exactly how much power a neighborhood will need, or exactly how many visits a particular website will receive on any given day. Since information about the demand is often necessary to develop optimal or near-optimal network routing solutions, network routing algorithms often need a way of obtaining or estimating future demand.
With the advent of big data and internet-of-things systems, crowd-sourcing has gained popularity as a demand forecasting approach for network routing systems. In crowd-sourcing, customers submit their request history to the network operator. The network operator uses this historical data to train a statistical or machine learning model to predict future demand from historical demand. 

While crowd-sourcing provides a bountiful supply of data for training demand forecasting models, it can also introduce potential privacy concerns. Since crowd-sourcing uses individual-level customer data to train demand forecasting models, the model's outputs may reveal sensitive user information, especially if it overfits to its training data \cite{CarliniEtAl21}. Such privacy risks are problematic because they may deter users from sharing their data with network operators, hence reducing the supply of training data for demand forecasting models. 

To address such concerns, the demand forecasting pipeline should be augmented with privacy-preserving mechanisms. Differential privacy \cite{DworkMNS06} is a principled and popular method to occlude the influence a single user's data on the result of a population study while also maintaining the study's accuracy. This is done by carefully injecting noise into the desired computation so that data sets that differ by at most one data point will produce statistically indistinguishable results. 

Providing differential privacy guarantees for the standard formulation of network routing is difficult because the constraints contain user data, meaning that in general feasibility and privacy become mutually exclusive. More specifically, in the standard network routing problem, the demand sources and sinks can be identified by checking for conservation of flow, and as a result, the presence of a user going to or from a very rare location can be detected from any feasible solution. Because differential privacy requires that the presence of any single user's data be difficult to detect from the algorithm's output, privacy and feasibility are at odds with one another in the standard formulation. 

\subsection{Statement of Contributions}

In this paper we present a differentially private algorithm for network routing problems. The main ingredient is a reformulation of network routing which moves all user data dependent parameters out of the constraint set and into the objective function. We then present an algorithm for solving this formulation based on differentially private variants of stochastic gradient descent. In this algorithm, differential privacy is achieved by injecting noise, and one may wonder if this noise injection compromises solution quality. We prove that our algorithm is differentially private and under several reasonable regularity conditions, is also asymptotically optimal (as the size of the training set goes to infinity). We note that in exchange for becoming compatible with differentially private algorithms, this new formulation is more computationally expensive.

\subsection{Related Work}

Traffic assignment in transportation systems is one of the most well-known applications of network routing. Herein we focus our literature review on privacy research in transportation networks. Privacy works in transportation mainly focus on \emph{location privacy}, where the objectives is to prevent untrusted and/or external entities from learning geographic locations or location sequences of an individual \cite{BeresfordStajano2003}. Privacy-preserving approaches have been proposed for various location-based applications, e.g.,  trajectory publishing, mobile crowdsensing, traffic control, etc. These techniques are based on spatial cloaking~\cite{ChowMokbelEtAl2011} and differential privacy~\cite{Dwork2008}. While not the setting of interest for this paper, there are many works that use Secure Multi-Party Computation (MPC) \cite{GoldreichMW87} to achieve privacy in \textit{distributed} mobility systems.

Spatial cloaking approaches aggregate users' exact locations into coarse information. These approaches are often based on $k$-anonymity \cite{Sweeney2002}, where a mobility dataset is divided into equivalence classes based on data attributes (e.g., geological regions, time, etc.) so that each class contains at least $k$ records \cite{GhasemzadehFungEtAl2014,HeChow2020}. These $k$-anonymity-based approaches can guarantee that every record in the dataset is indistinguishable from at least $k-1$ other records. However, $k$-anonymity is generally considered to be a weak privacy guarantee, especially when $k$ is small. Furthermore, very coarse data aggregation is required to address outliers or sparse data, and in these cases spatial cloaking-based approaches provide low data accuracy. 

Differential privacy provides a principled privacy guarantee by producing randomized responses to queries, whereby two datasets that differ in only one entry produce statistically indistinguishable responses \cite{DworkMNS06}. In other words, differential privacy ensures that an adversary with arbitrary background information (e.g., query responses, other entries) cannot infer individual entries with high confidence. Within transportation research, \cite{WangHuEtAl2018,YanLuEtAl2019} share noisy versions of location data for mobile crowd-sensing applications. \cite{GongZhangEtAl2015,GursoyLiuEtAl2018,HussaeniFungEtAl2018,LiYangEtAl2020} use differential privacy to publish noisy versions of trajectory data. \cite{DongEtAl15} and \cite{HanEtAl17} apply differential privacy to gradient descent algorithms for federated learning in mobility systems.

Many of the works mentioned in the previous paragraph establish differential privacy of their proposed algorithms by using composition properties of differential privacy. Composition theorems for differential privacy describe how well privacy is preserved when conducting several computations one after another. In \cite{DongEtAl15} and \cite{HanEtAl17}, composition theorems are applied as black boxes without considering the mathematical properties of the gradient descent algorithm. As a result, the privacy guarantees are overly conservative, meaning that large amounts of noise are added to the algorithm, leading to suboptimal behavior both in theory and in practice. Similarly, \cite{GongZhangEtAl2015,GursoyLiuEtAl2018,HussaeniFungEtAl2018,LiYangEtAl2020} use composition rules as a black box, and while privacy is achieved in this way, there are no accuracy guarantees for the algorithms presented in those works. 

While blackbox applications of differential privacy can lead to impractical levels of noise injection, specialized applications of differential privacy were discovered that could provide privacy with much less noise. \cite{WuEtAl17} show how a simple adjustment to stochastic gradient descent can give rise to an algorithm which is both differentially private, and under reasonable regularity assumptions, is also asymptotically optimal. \cite{FeldmanMTT18} and \cite{FeldmanKT20} refined this idea to develop stochastic gradient descent algorithms that are differentially private, computationally efficient, and have optimal convergence rates. These techniques cannot directly be used to solve the standard formulation of network routing because they study unconstrained optimization problems or optimization problems with public constraint sets (i.e., constraints that do not include any private data). 

\section{Model}\label{sec:model}

In this section we define notations, network models, and assumptions that will allow us to formulate network routing as a data-driven convex optimization problem. 



\subsection{Preliminaries}\label{sec:model:prelims}

\noindent \textbf{Indicator Representation of Edge Sets:} Let $G = (V,E)$ be a graph with vertices $V$ and edges $E = \bigbrace{e_1,...,e_m}$. For any subset of edges $E' \subset E$, we define the indicator representation of $E'$ as $\mathds{1}_{[E']}$ as a boolean vector of length $m$ in the following way: The $i$th entry of $\mathds{1}_{[E']}$ is $1$ if $e_i \in E'$, and is $0$ otherwise. \\

\noindent \textbf{Derivative Notation:} For a scalar valued function $ x \mapsto f(x)$, we use $\nabla f(x_0)$ to denote the gradient of the $f$ with respect to $x$ evaluated at the point $x = x_0$. For a vector valued function $x \mapsto g(x)$, and a variable $z$, we use $\cD_{z}[g](x_0)$ to denote the derivative matrix of $g$ with respect to $z$ evaluated at the point $x = x_0$. \\

\noindent \textbf{Projections:} For a convex set $\cS \subset \mathbb{R}^d$, we use $\Pi_{\cS} : \mathbb{R}^d \rightarrow \mathbb{R}^d$ to denote the projection operator for $\cS$. For any $x \in \mathbb{R}^d$, $\Pi_S(x) := \arg\min_{s \in \cS} \norm{x-s}_2$ to be the point in $S$ that is closest to $x$. 

\subsection{Network, Demand, and Network Flow}\label{sec:model:network}

In this section we will introduce a) a graph model of a network, b) a representation of network demand, c) the standard formulation for network routing and d) privacy requirements. The notation defined in this section is aggregated in table form in Section~\ref{app:notation} for the reader's convenience. 

\begin{definition}[Network Representation]\label{def:graph_representation}
We use a directed graph $G = (V,E)$ to represent the network, where $V$ and $E$ represent the set of vertices and edges in the network, respectively. We will use $n := \abs{V}$ and $m := \abs{E}$ to denote the number of vertices and edges in the graph, respectively. For vertex pairs $(o,d) \in V \times V$ we will use $\cP_G(o,d)$ to denote the set of simple paths from $o$ to $d$ in $G$. 
\end{definition}

\begin{definition}[Operation Period]\label{def:operation-period}
We use $\cT := \bigbra{t_{\text{start}}, t_{\text{end}}}$ to denote the operation period within which a network operator wants to optimize its routing decisions. We will also use $T$ to denote the number of minutes in the operation period. For example, $t_{\text{start}} = 8:00 \text{am}, \; t_{\text{end}} = 9:00 \text{am}$ could represent a morning commute period where $T = 60$. 
\end{definition}

\begin{definition}[Demand Representation]\label{def:demand_matrix}
We study a stationary demand model where demand within the operation period $\cT$ is specified by a matrix $\Lambda \in \mathbb{R}_+^{n \times n}$. For each ordered pair of vertices $(o,d) \in V \times V$, $\Lambda(o,d)$ is the number of requests arriving per minute (i.e., the arrival rate) during $\cT$ that need routing from vertex $o$ to vertex $d$.
\end{definition}

\begin{remark}[Estimating $\Lambda$ from historical data]\label{rem:lambda_from_data}
The arrival rates from historical demand are computed empirically, i.e., if $\Lambda_t$ represents the demand for day $t$, then $\Lambda_t(o,d)$ is calculated by counting the number of $(o,d)$ requests appearing on day $t$, and then dividing it by $T$ to obtain requests per minute.
\end{remark}

\begin{definition}[Link Latency Functions]\label{def:delay_functions}
To model congestion effects, each link $e \in E$ has a latency function $f_e : \mathbb{R}_+ \rightarrow \mathbb{R}_+$ which specifies the average travel time through the link as a function of the total flow on the link.
\end{definition}

In this paper we study a setting where a network operator wants to route demand while minimizing the total travel time for the requests. With these definitions, the standard formulation of minimum travle time network routing is described in Definition~\ref{def:standard_nf}.

\begin{definition}[Standard Formulation of Network Flow]\label{def:standard_nf}
For a network $G = (V,E)$ with link latency functions $\bigbrace{f_e}_{e \in E}$ and demand $\Lambda$, the standard network flow problem is the following minimization program:
\begin{subequations}\label{eqn:standard_nf}
\begin{align}
    \underset{x}{\text{minimize}} & \sum_{e \in E} y_e f_e(y_e) \label{eqn:standard_nf:obj}\\
    \text{subject to } & x = \bigbrace{x^{(o,d)}}_{(o,d) \in V \times V}, \label{eqn:standard_nf:x}\\ 
    y_e &= \sum_{(o,d) \in V \times V} x_e^{(o,d)} \text{ for all } e \in E, \label{eqn:standard_nf:total_flow}\\
    \sum_{v : (u,v) \in E} x_{(u,v)}^{(o,d)} - \sum_{w : (w,u) \in E} x_{(w,u)}^{(o,d)} &= \Lambda(o,d) \mathds{1}_{[u = o]} - \Lambda(o,d) \mathds{1}_{[u = d]} \text{ for all } u \in V, (o,d) \in V \times V. \label{eqn:standard_nf:flow_conserve}
\end{align}
\end{subequations}
\end{definition}
In \eqref{eqn:standard_nf}, the decision variable $x$ is a collection of flows $\bigbrace{x^{(o,d)}}_{(o,d) \in V \times V}$, one for each non-zero entry of $\Lambda$, as represented by constraint \eqref{eqn:standard_nf:x}. \eqref{eqn:standard_nf:flow_conserve} are flow conservation constraints to ensure that $x^{(o,d)}$ sends $\Lambda(o,d)$ units of flow from $o$ to $d$. Constraints \eqref{eqn:standard_nf:total_flow} ensure that $\bigbrace{y_e}_{e \in E}$ represents the total amount of flow on each edge. Finally, the objective function \eqref{eqn:standard_nf:obj} is to minimize the total travel time as a function of total network flow. 

In the next subsection we will describe the rigorous privacy requirements that we will mandate while designing algorithms for network flow. We then describe in Section~\ref{sec:model:feas_vs_priv} why privacy and feasibility are mutually exclusive in the standard network flow formulation.

\subsection{Privacy Requirements}\label{sec:model:privacy_req}

We will use differential privacy to reason about the privacy-preserving properties of our algorithms. At a high level, changing one data point of the input to a differentially private algorithm should lead to a statistically indistinguishable change in the output. To make this concept concrete we will need to define data sets and adjacency.

\begin{definition}[Data sets]
Given a space of data points $\cZ$, a data set $L$ is any finite set of elements from $\cZ$. In practice, each element of a data set is data collected from a single user, or data collected during a specific time period. We will use $\cL$ to denote the set of all data sets. 
\end{definition}

\begin{definition}[Data Set Adjacency]
Given a space of data sets $\cL$, an adjacency relation $\text{Adj}$ is a mapping $\text{Adj}: \cL \times \cL \rightarrow \bigbrace{0 , 1}$. Two data sets $L_1, L_2$ are said to be adjacent if and only if $\text{Adj}(L_1, L_2) = 1$. 
\end{definition}

While the exact definition can vary across applications, adjacent data sets are sets that are very similar to one another. The most common definition of adjacency is the following: $L,L'$ are adjacent if and only if $L'$ can be obtained by adding, deleting, or modifying at most one data point from $L$, and vice versa. Thus comparing a function's output on adjacent data sets measures how sensitive the function is to changes in a single data point. 

With these definitions in place, we are now ready to define differential privacy. 

\begin{definition}[Differential Privacy]\label{def:diffpriv}
For a given adjacency relation $\text{Adj}$, a function $M : \cL \rightarrow \cX$ is $(\epsilon,\delta)$-differentially private if for any $L_1, L_2 \in \cL$ with $\text{Adj}(L_1, L_2) = 1$, the following inequality
\begin{align*}
    \mathbb{P} \bigbra{ M(L_1) \in E } \leq e^\epsilon \mathbb{P} \bigbra{ M(L_2) \in E } + \delta 
\end{align*}
holds for every event $E \subset \cX$. 
\end{definition}

The definition of differential privacy requires that changing one data point in the input to a differentially private algorithm cannot change the distribution of the algorithm's output by too much. Such a requirement ensures the following strong privacy guarantee: It is statistically impossible to reliably infer the properties of a single data point just by examining the algorithm's output, even if the observer knows all of the other data points in the input \cite{DworkMNS06}. 

To proceed, we must specify what data set adjacency means in the context of network routing. For network routing problems, historical data is often a collection of routing requests through the network. To protect the privacy of those who submit requests, we want the routing policy we compute to not depend too much on any single request that appears in the historical data. This motivates the the following notion of data set adjacency that we will be using throughout this paper. 

\begin{definition}[Request Level Adjacency (RLA)]\label{def:trip_level_adjacency}
We define a function $\text{RLA} : \cL \times \cL \rightarrow \bigbrace{0,1}$ which maps pairs of data sets to booleans. For two historical datasets of network demand $L := (\Lambda_1,...,\Lambda_N)$ and $L' := (\Lambda'_1,...,\Lambda'_N)$, we say $L$ and $L'$ are request-level-adjacent (RLA) if exactly one of the following statements is true:
\begin{enumerate}
    \item $L$ contains all of the requests in $L'$, and contains one extra request that is not present in $L'$.
    \item $L'$ contains all of the requests in $L$, and contains exactly one extra request that is not present in $L$.
\end{enumerate}
Mathematically, $L,L'$ are request-level-adjacent, i.e., $\text{RLA}(L, L') = 1$, if and only if they satisfy all of the following relations:
\begin{itemize}
    \item There exists $t$ so that $\Lambda_k = \Lambda'_k$ for all $k \neq t$.
    \item There exists two vertices $o$ and $d$ so that $\Lambda_t(o',d') = \Lambda'_t(o',d')$ for all $(o',d') \neq (o,d)$. 
    \item $\abs{\Lambda_t(o,d) - \Lambda'_t(o,d)} \leq \frac{1}{T}$.
\end{itemize}
Indeed, these relations dictate that one of the datasets contains an extra request from $o$ to $d$ which happened on the $t$th day. A difference of $1$ request within a $T$ minute operation period leads to a change of $\frac{1}{T}$ in the arrival rate. Aside from this difference, the datasets are otherwise identical. 
\end{definition}

\subsection{Differential Privacy Challenges in Standard Network Flow}\label{sec:model:feas_vs_priv}

In the introduction we mentioned that privacy and feasibility can be mutually exclusive in the standard formulation of network flow described in \eqref{eqn:standard_nf}. In this section we formally show that if $\Lambda$ is constructed from a data set of trips as described in Remark~\ref{rem:lambda_from_data}, then trips to or from uncommon locations can be easily detected from any feasible solution to \eqref{eqn:standard_nf}. As a result, announcing or releasing a feasible solution to \eqref{eqn:standard_nf} is not, in general, differentially private. Formally, we will prove the following theorem in this section.

\begin{theorem}[Differential Privacy Impossibility for Standard Network Flow]\label{thm:standard_privacy_impossible}
Let $M$ be a function that takes as input a matrix $\Lambda$ with non-negative entries and returns a feasible solution to the optimization problem $\eqref{eqn:standard_nf}$ where $\Lambda$ is used as a demand matrix. Then $M$ cannot be $(\epsilon,\delta)$-differentially private for any $\delta < 1$.
\end{theorem}

We further note that $(\epsilon, \delta)$-differential privacy only provides a meaningful privacy guarantee when $\epsilon < 1$ and $\delta < 1$ \cite{Dwork2008}. 

\begin{proof}[Proof of Theorem~\ref{thm:standard_privacy_impossible}]
Let $G = (V,E)$ be a network, and $\Lambda$ be constructed from a historical data set of requests in $G$. Suppose there exist uncommon locations $o$ and $d$ for which $\Lambda$ contains no trips to or from either $o$ or $d$. Mathematically, this means that

$$\Lambda(o, u) = 0, \Lambda(u,o) = 0, \Lambda(d,u) = 0, \Lambda(u,d) = 0 \text{ for all } u \in V. $$

Such situations are not uncommon in transportation networks, if, for example, $o$ and $d$ are the homes of two different people who do not drive (perhaps they walk or bike to and from work). 

If we now add a trip from $o$ to $d$ to the data set, and let $\Lambda'$ be the resulting demand matrix, then we have 
\begin{enumerate}
    \item[i] $\Lambda(o',d') = \Lambda'(o',d')$ for all $(o',d') \neq (o,d)$, and
    \item[ii] $\Lambda(o,d) = 0$, $\Lambda'(o,d) = \frac{1}{T}$. 
\end{enumerate}

Let $\text{Prob}_1, \text{Prob}_2$ be the optimization problem $\eqref{eqn:standard_nf}$ using demand $\Lambda, \Lambda'$ respectively. Because $\Lambda,\Lambda'$ are request-level-adjacent, any differentially private algorithm must behave similarly when acting on $\text{Prob}_1 \text{ and } \text{Prob}_2$. However, this is impossible because the feasible sets of $\text{Prob}_1, \text{Prob}_2$ are disjoint. If we look at constraint $\eqref{eqn:standard_nf:flow_conserve}$ with $u = d$ and $(o,d)$ then any feasible solution to $\text{Prob}_1$ satisfies
\begin{align*}
    \sum_{u : (u,d) \in E} x_{(u,d)}^{(o,d)} - \sum_{w : (d,w) \in E} x_{(u,d)}^{(o,d)} = \Lambda(o,d) \mathds{1}_{[d = d]} = 0.
\end{align*}
However, any feasible solution to $\text{Prob}_2$ satisfies
\begin{align*}
    \sum_{u : (u,d) \in E} x_{(u,d)}^{(o,d)} - \sum_{w : (d,w) \in E} x_{(u,d)}^{(o,d)} = \Lambda(o,d) \mathds{1}_{[d = d]} = \frac{1}{T}.
\end{align*}
In other words, checking the net flow leaving node $d$ will detect the presence of any trips going to or from $d$. We will now show that any algorithm which returns a feasible solution to \eqref{eqn:standard_nf} cannot be differentially private. To this end, define the event $E$ to be the event that flow is conserved at node $d$. Then for any algorithm $M$ that takes as input a demand matrix and returns a feasible solution to \eqref{eqn:standard_nf} with the specified demand matrix, we have $\mathbb{P}[M(\Lambda) \in E] = 1, \mathbb{P}[M(\Lambda') \in E] = 0$. Recalling Definition~\ref{def:diffpriv}, $M$ is $(\epsilon,\delta)$-differentially private only if $\mathbb{P}[M(\Lambda) \in E] \leq e^\epsilon \mathbb{P}[M(\Lambda') \in E] + \delta$. This equation can only be satisfied if $\delta \geq 1$. 
\end{proof}

We have two remarks regarding Theorem~\ref{thm:standard_privacy_impossible}.

\begin{remark}
The same result holds if $M$ returns the total flow $y$ associated with a feasible solution (see \eqref{eqn:standard_nf:total_flow}), rather than returning the feasible solution itself. In other words, even total traffic measurements (without knowing the breakdown by $(o,d)$ pairs) can already expose trips to or from uncommon locations.
\end{remark}

\begin{remark}
The vulnerability of trips to and from uncommon locations is not a purely theoretical concern. A study on the New York City Taxi and Limousine data set was able to identify trips from residential areas to gentlemen's club \cite{Pandurangan14}. Because the start locations were in residential areas, it was easy to re-identify the users who booked the taxi trips as the owners of the homes that the taxi trips began at. As a result, despite the data set being anonymized, users who had taken taxis to this gentlemen's club were re-identified, and their reputations may have been negatively affected.  
\end{remark}


\section{Routing Policy Formulation of Network Flow}\label{sec:netflow_reformulate}

To sidestep the impossibility result described in Theorem~\ref{thm:standard_privacy_impossible}, we present an alternative formulation for the network flow problem in this section. The alternative formulation avoids the issues mentioned in Theorem~\ref{thm:standard_privacy_impossible} by moving all parameters related to user data from the constraints to the objective function, as described in \eqref{eqn:opt:deterministic:omniscient}. We note that \eqref{eqn:opt:deterministic:omniscient} can only be solved if the demand $\Lambda$ is known, which may not always be the case. For this reason, we present two variations of \eqref{eqn:opt:deterministic:omniscient}: \eqref{eqn:opt:random:omniscient} is the stochastic version of \eqref{eqn:opt:deterministic:omniscient} where $\Lambda$ is drawn from a distribution $\cQ$, and \eqref{eqn:opt:random:datadriven} is the data driven approximation to \eqref{eqn:opt:random:omniscient} that one would solve if $\cQ$ is unknown. 

Before formally defining the model, we provide a high level description of how this formulation works. In this formulation, a feasible solution $x = \bigbrace{ x^{(o,d)} }_{(o,d) \in V \times V}$ specifies, for each $(o,d) \in V \times V$, a flow $x^{(o,d)}$ that routes 1 unit of flow from $o$ to $d$. We note that a flow is specified for $(o,d)$ even if there is no demand for this origin and destination in $\Lambda$, i.e., $\Lambda(o,d) = 0$. We refer to $x$ as a \textit{routing policy} due to its connection to randomized routing, which Remark~\ref{rem:total_net_flow} discusses in further detail. Given a feasible solution $x$, the objective function first calculates the total traffic on each edge by taking a linear combination of $\bigbrace{ x^{(o,d)} }_{(o,d) \in V \times V}$ flows, where the coefficients of the linear combination are determined by the demand $\Lambda$, with high demand $(o,d)$ pairs having larger coefficients. The total travel time can be computed from the total traffic in the same way as \eqref{eqn:standard_nf:obj}. These ideas are formalized by the following definitions. 

\begin{definition}[Unit $(o,d)$ flow]\label{def:unit_od_flow}
For a given origin-destination pair $(o,d)$, we say that a flow $x^{(o,d)} \in \mathbb{R}_+^m$ is a unit $(o,d)$ flow if and only if it routes exactly $1$ unit of flow from $o$ to $d$. Formally, this condition is represented by the following constraints:
\begin{align}
\sum_{v \in V : (v,u) \in E} x^{(o,d)}_{(v,u)} - \sum_{v \in V : (u,v) \in E} x^{(o,d)}_{(u,v)} = 
\casewise{
\begin{tabular}{cc}
$-1$ & if $u = o$ \\
$1$ & if $u = d$ \\
$0$ & otherwise.
\end{tabular}
} \text{ for all } u \in V \label{eqn:od_flow_conservation}
\end{align}
Indeed, \eqref{eqn:od_flow_conservation} requires that the net flow entering $o$ is $-1$, the net flow entering $d$ is $1$, and that flow is conserved at all other vertices in the graph. 
\end{definition}

\begin{definition}[Unit Network Flow]\label{def:unit_network_flow}
A unit network flow is a collection of flows $x = \bigbrace{x^{(o,d)}}_{(o,d) \in V \times V}$ so that $x^{(o,d)}$ is a unit $(o,d)$ flow for each $(o,d) \in V \times V$. We use $\cX$ to denote the set of all unit network flows.
\end{definition}

\begin{remark}\label{rem:x_stack}
We can represent $x$ as a concatenation of the vectors $\bigbrace{x^{(o,d)}}_{(o,d) \in V \times V}$. Since each unit $(o,d)$ flow is a $m$ dimensional vector, we have $x \in \mathbb{R}^{n^2 m}$. 
\end{remark}

We will refer to unit network flows as routing policies, due to their connection with randomized routing, described in the following remark.

\begin{remark}[Routing demand using unit flow policies]\label{rem:total_net_flow}
A unit network flow represents a randomized routing policy. Due to the flow decomposition theorem \cite{TrevisanNotes}, every unit $(o,d)$ flow can be written as a distribution of paths from $o$ to $d$. Formally, for any unit $(o,d)$ flow $x^{(o,d)}$, there exist paths $p_1^{(o,d)}, p_2^{(o,d)}, ..., p_m^{(o,d)}$ from $o$ to $d$ and weights $w_1^{(o,d)}, ..., w_m^{(o,d)}$ so that the following equations hold:
\begin{align}
    x^{(o,d)} &= \sum_{i=1}^m w_i^{(o,d)} \mathds{1}_{[p_i^{(o,d)}]} \label{eqn:flow_decomp}\\
    &\sum_{i=1}^m w_i^{(o,d)} = 1 \nonumber \\
    & w_i^{(o,d)} \geq 0 \text{ for all } 1 \leq i \leq m. \nonumber  
\end{align}
Furthermore, $p_1^{(o,d)},...,p_m^{(o,d)}$ and $w_1^{(o,d)},...,w_m^{(o,d)}$ are efficiently computable. Defining the probability distribution $P_{x^{(o,d)}}$
\begin{align*}
    P_{x^{(o,d)}}(\mathds{1}_{[p_i^{(o,d)}]}) = w_i^{(o,d)} \text{ for all } 1 \leq i \leq m,
\end{align*}
$x^{(o,d)}_e$ represents the probability that a random path chosen from $P_{x^{(o,d)}}$ contains $e$.
$x$ describes the expected behavior of the randomized routing policy that determines routes for $(o,d)$ requests by drawing a path independently at random from $P_{x^{(o,d)}}$. In particular, when using this policy to serve demand $\Lambda$, by linearity of expectation, the expected number of requests from $o$ to $d$ whose assigned path contains $e$ is exactly $\Lambda(o,d) x_e^{(o,d)}$. Furthermore, the average number of requests on each edge will be $\sum_{(o,d) \in V \times V} \Lambda(o,d) x^{(o,d)}$. 
\end{remark}

\subsection{Minimum Total Travel Time Network Flow}\label{sec:model:det_omni_opt}

In the minimum travel time network flow problem, the network operator wants to find a stationary routing policy for each $(o,d)$ pair that will lead to small total travel times for the requests. Due to the equivalence between stationary $(o,d)$ routing policies and unit $(o,d)$ flows, the network operator can instead search over the space of unit $(o,d)$ flows. 

The total travel time of a flow $y$ through $G$ is given by $\sum_{e \in E} y_e f_e(y_e)$. The total flow resulting from a unit network flow $x$ serving demand $\Lambda$ is the sum of the flow contributions from each of the $(o,d)$ pairs. With Remark~\ref{rem:total_net_flow} in mind, the total amount of flow on an edge $e \in E$ when serving $\Lambda$ according to $x$ is given by $\sum_{(o,d) \in V \times V} \Lambda(o,d) x_e^{(o,d)}$. We can thus define $F(x,\Lambda)$, the total travel time experienced by requests $\Lambda$ when being routed by $x$, as follows:

\begin{align}
    F(x, \Lambda) := \sum_{e \in E} \bigpar{ \sum_{(o,d) \in V \times V} \Lambda(o,d) x_e^{(o,d)} } f_e \bigpar{ \sum_{(o,d) \in V \times V} \Lambda(o,d) x_e^{(o,d)} }. \label{eqn:network_flow_to_time}
\end{align}

With these definitions in place, the unit network flow that minimizes the total travel time when serving the demand $\Lambda$ can be found by solving the following optimization problem

\begin{align}
	\underset{x}{\text{minimize }} & F(x,\Lambda) \label{eqn:opt:deterministic:omniscient}\\
	\text{subject to } & x = \bigbrace{x^{(o,d)}}_{(o,d) \in V \times V}, \nonumber \\
	& x^{(o,d)} \text{ is a unit $(o,d)$ flow for every} (o,d) \in V \times V, \nonumber 
\end{align}

\subsection{Network Flow with Stochastic Demand}\label{sec:model:demand_model}

In practice, demand may vary from day to day, and such variations can be modeled by $\Lambda$ being a random variable with distribution $\cQ$. If $\cQ$ is known by the network operator, then rather than solving \eqref{eqn:opt:deterministic:omniscient}, the operator is interested in minimizing expected total travel time through the following stochastic optimization problem:

\begin{align}
	\underset{x}{\text{minimize }} & \mathbb{E}_{\Lambda \sim \cQ} \bigbra {F(x,\Lambda)} \label{eqn:opt:random:omniscient}\\
	\text{subject to } & x = \bigbrace{x^{(o,d)}}_{(o,d) \in V \times V}, \nonumber \\
	& x^{(o,d)} \text{ is a unit $(o,d)$ flow for every} (o,d) \in V \times V, \nonumber 
\end{align}

We note that \eqref{eqn:opt:random:omniscient} is a generalization of \eqref{eqn:standard_nf} to the case when $\Lambda$ is random. 

In the more realistic case that $\cQ$ is not known, the optimization problem \eqref{eqn:opt:random:omniscient} can be approximated from historical data. We study a situation where the network operator has demand data $\Lambda_1, \Lambda_2, ..., \Lambda_N \iid \cQ$ collected from operations of previous days. Using this data it can solve the following empirical approximation to \eqref{eqn:opt:random:omniscient}:

\begin{align}
	\underset{x}{\text{minimize }} & \frac{1}{N} \sum_{k=1}^N F(x,\Lambda_k) \label{eqn:opt:random:datadriven}\\
	\text{subject to } & x = \bigbrace{x^{(o,d)}}_{(o,d) \in V \times V}, \nonumber \\
	& x^{(o,d)} \text{ is a unit $(o,d)$ flow for every} (o,d) \in V \times V, \nonumber 
\end{align}

The optimization problem in \eqref{eqn:opt:random:datadriven} uses historical data to estimate \eqref{eqn:opt:random:omniscient}. In line with Assumption~\ref{assump:bounded-demand} we will assume that $\Lambda_t(o,d) \leq \lambda_{\max}$ for all values of $t, o$ and $d$. 

In Section~\ref{sec:algorithm} we show how \eqref{eqn:opt:random:datadriven} can be solved in a request-level differentially private way. 

\subsection{Assumptions on Travel Time functions}\label{sec:model:assumptions}

In this section we will introduce some assumptions that will help us establish our technical results. We will make the following assumptions on the network demand:

\begin{assumption}[Bounded Demand]\label{assump:bounded-demand}
We assume there exists a non-negative constant $\lambda_{\max}$ so that $\Lambda(o,d) \leq \lambda_{\max}$ for every $(o,d) \in V \times V$. In practice, this constant can be estimated from historical data. 
\end{assumption}

\noindent The following are assumptions we make on the objective function $F$ (see \eqref{eqn:network_flow_to_time}). These assumptions are related to properties of the link latency functions $\bigbrace{f_e}_{e \in E}$. We present a typical model for link latency functions in Section~\ref{sec:model:transit-net-ex} that satisfies all of the following assumptions. 

\begin{assumption}[Bounded Variance Gradients]\label{assump:bounded-var-grads} We assume there exists a non-negative constant $K$ so that for every $x$, the variance of $\nabla F(x,\Lambda)$ is upper bounded by $K^2$, i.e., $\mathbb{E}_{\Lambda \sim \cQ} \bigbra{ \norm{\nabla F(x,\Lambda)}_2^2 } \leq K^2$. 
\end{assumption}

\begin{assumption}[Twice Differentiability]\label{assump:twice_diff}
We assume that $F(\cdot, \Lambda)$ is twice-differentiable for every $\Lambda$ so that the hessian $H(x, \Lambda) := \frac{\partial^2}{\partial x^2} F(x,\Lambda)$ is defined on the entire domain of $x$. 
\end{assumption}

\begin{assumption}[Strong Convexity]\label{assump:strong_convex}
We assume that there exists $\alpha > 0$ for which $F(\cdot, \Lambda)$ is $\alpha$-strongly convex for every $\Lambda$. This means that for every $\Lambda$, and any unit network flows $x,x'$ we have
\begin{align*}
    F(x',\Lambda) \geq F(x, \Lambda) + \nabla_x F(x, \Lambda)^\top (x' - x) + \frac{\alpha}{2} \norm{x'-x}_2^2
\end{align*}
\end{assumption}

\begin{assumption}[Smoothness]\label{assump:smoothness}
We assume that there exists $\beta > \alpha$ for which $F(\cdot, \Lambda)$ is $\beta$-smooth for every $\Lambda$. This means that for every $\Lambda$, and any unit network flows $x,x'$ we have
\begin{align*}
    F(x',\Lambda) \leq F(x, \Lambda) + \nabla_x F(x, \Lambda)^\top (x' - x) + \frac{\beta}{2} \norm{x'-x}_2^2.
\end{align*}
\end{assumption}

\begin{assumption}[Bounded second order partial derivative]\label{assump:bounded_double_derivative}
We assume that there exists $C > 0$ so that $\norm{\cD_{\Lambda(o,d)} \bigbra{ \nabla_x F}(x,\Lambda)}_{op} \leq C$ for all $x,\Lambda$ and $(o,d) \in V \times V$.  
\end{assumption}

\begin{remark}[Satisfying Assumption \ref{assump:bounded_double_derivative}]
We can satisfy Assumption~\ref{assump:bounded_double_derivative} for any positive $C$ by re-scaling. In particular, letting $\theta := \max_{x', \Lambda'} \norm{\cD_{\Lambda(o,d)} \bigbra{ \nabla_x F}(x',\Lambda')}_2$, then for any $C$ we can define a re-scaled objective function 
\begin{align*}
    \widetilde{F}(x, \Lambda) := \frac{C}{\theta} F(x,\Lambda).
\end{align*}
By construction, $\widetilde{F}$ satisfies Assumption~\ref{assump:bounded_double_derivative} with constant value $C$. Note, however, that the smoothness and strong convexity parameters for $\widetilde{F}$ will be rescaled accordingly to $\frac{\beta}{\theta}$ and $\frac{\alpha}{\theta}$ respectively. 
\end{remark}

\subsection{Transportation Model satisfying assumptions from Section~\ref{sec:model:assumptions}}\label{sec:model:transit-net-ex}


In this section, we present a transportation network model that satisfies Assumptions~\ref{assump:twice_diff},\ref{assump:strong_convex},\ref{assump:smoothness} and \ref{assump:bounded_double_derivative}. We study a network where the link latency functions are all affine where for each $e \in E$, there are non-negative constants $q_e, c_e$ so that $f_e(y) = q_e y + c_e$. Let $Q \in \mathbb{R}^{m \times m}$ be defined so that $Q_{ee} = q_e$, and let $c \in \mathbb{R}^m$ be the concatenation of all of the zero order coefficients in the link latency functions. 

As mentioned in Remark~\ref{rem:x_stack}, we will represent $x$ as a concatenation of $\bigbrace{x^{(o,d)}}_{(o,d) \in V \times V}$. As such, $x \in \mathbb{R}^{n^2m}$. Let $\bigbrace{(o_i,d_i)}_{i=1}^{n^2}$ be the order in which the unit $(o,d)$ flows are concatenated to produce $x$ so that

\begin{align*}
    x = \bigbra{
        \begin{tabular}{c}
            $x^{(o_1,d_1)}$\\
            $x^{(o_2,d_2)}$ \\
            $\vdots$ \\
            $x^{(o_{n^2}, d_{n^2})}$
        \end{tabular}
    }.
\end{align*}

The total flow on the links in the network when serving demand $\Lambda$ according to $x$ can then be written as:
\begin{align*}
    y := \sum_{(o,d) \in V \times V} \Lambda(o,d) x^{(o,d)} = B_\Lambda x
\end{align*}
where $B_\Lambda = \text{vec}(\Lambda) \otimes I_m$. Here $\text{vec}(\Lambda)$ is a $n^2$ dimensional row vector whose $i$th entry is $\Lambda(o_i,d_i)$, and $\otimes$ represents the Kronecker product. Then when $\Lambda$ is being routed according to $x$, the travel times on the links can be computed as
\begin{align*}
    Q y + c = Q B_\Lambda x + c,
\end{align*}
which means that the total travel time can be written as
\begin{align}
    \widehat{F}(x,\Lambda) := \sum_{e \in E} y_e f_e(y_e) &= y^\top \bigpar{ Q y + c } \nonumber \\
    &= x^\top B_\Lambda^\top \bigpar{ Q B_\Lambda x + c} \nonumber \\
    &= x^\top B_\Lambda^\top Q B_\Lambda x + c^\top B_\Lambda x. \label{eqn:transit-ex:time}
\end{align}

If we add a bit of $\ell_2$ regularization, we obtain
\begin{align*}
    F(x,\Lambda) &= \widehat{F}(x,\Lambda) + \frac{\alpha}{2} \norm{x}_2^2\\
    &= x^\top B_\Lambda^\top Q B_\Lambda x + c^\top B_\Lambda x + \frac{\alpha}{2} \norm{x}_2^2 \\
    &= x^\top \bigpar{ B_\Lambda^\top Q B_\Lambda + \frac{\alpha}{2} I} x + c^\top B_\Lambda x. 
\end{align*}
Recall that we use $H(\cdot,\Lambda)$ to denote the hessian of $F(\cdot, \Lambda)$ with respect to its first argument. We now make the following observations:
\begin{itemize}
    \item The Hessian of $F(x,\Lambda)$ with respect to $x$ is defined for all $x$ and is equal to $ 2 B_\Lambda^\top Q B_\Lambda + \alpha I$. Hence Assumption~\ref{assump:twice_diff} is satisfied.
    \item Since $Q$ is a diagonal matrix with non-negative entries, $Q \succeq 0$. This implies that $B_\Lambda^\top Q B_\lambda \succeq 0$. Hence $H(x,\Lambda) \succeq \alpha I$. This implies that $F$ is $\alpha$-strongly convex, and hence Assumption~\ref{assump:strong_convex} is satisfied. 
    \item Note that $\norm{B_\Lambda^\top Q B_\Lambda + (\alpha/2) I}_{op} \leq \norm{Q}_{op} \norm{B_\Lambda}_{op}^2 + \alpha/2 \leq \norm{Q}_{op} \norm{\text{vec}(\Lambda)}_{2}^2 \norm{I_m}_{op}^2 + \alpha/2$. Defining $\beta := n^2 \bigpar{ \max_e q_e } \lambda_{\max}^2 + \alpha/2$, we see that $\norm{H(x,\Lambda)}_{op} \leq \beta$ for all $x \in \cX$, meaning that $F$ is $\beta$-smooth, and thus Assumption~\ref{assump:smoothness} is satisfied. 
    \item By product rule,
    \begin{align*}
        \cD_{\Lambda(o_i,d_i)} \bigbra{ \nabla F } (x, \Lambda) &= \frac{\partial}{\partial\Lambda(o_i,d_i)} \bigbrace{ 2\bigpar{B_\Lambda^\top Q B_\Lambda + \alpha I} x + B_\Lambda^\top c } \\
        &= \frac{\partial}{\partial \Lambda(o_i,d_i)} \bigbrace{ 2 B_\Lambda^\top Q y + 2 \alpha x + B_\Lambda^\top c } \\
        &= 2 \bigpar{ (e_i \otimes I) Q y + B_\Lambda^\top Q x^{(o,d)}} + e_i \otimes c
    \end{align*}
    where $e_i$ is the $i$th standard basis vector for $\mathbb{R}^{n^2}$. By triangle inequality we can then conclude that
    \begin{align*}
        \norm{\cD_{\Lambda(o_i,d_i)} \bigbra{ \nabla F } (x, \Lambda)}_2 &\leq 2 \norm{ Qy }_2 + 2 \norm{B_\Lambda^\top Q x^{(o,d)}}_2 + \norm{c}_2 \\
        &\leq 2 \norm{Q}_{op} \norm{y}_2 + 2 \norm{B_\Lambda}_{op} \norm{Q}_{op} \norm{x^{(o,d)}}_2 + \norm{c}_2 \\
        &\overset{(a)}{\leq} 2 \norm{Q}_{op} \norm{y}_2 + 2 n \lambda_{\max} \norm{Q}_{op} \sqrt{m} + \norm{c}_2 \\
        &\overset{(b)}{\leq} 2 \lambda_{\max} \norm{Q}_{op} n^2 \sqrt{m} + 2 \lambda_{\max} \norm{Q}_{op} n \sqrt{m} + \norm{c}_2. 
    \end{align*}
    Here $(a)$ is due to the fact that $x^{(o,d)}$ is a unit $(o,d)$ flow, meaning $\norm{x^{(o,d)}}_{\infty} \leq 1$ and thus $\norm{x^{(o,d)}}_2 \leq \sqrt{m}$. Also, $B_\Lambda = \text{vec}(\Lambda) \otimes I_m$ implies that $\norm{B_\Lambda}_{op} \leq \norm{\mathds{1}_m^\top}_2 \norm{I_m}_{op} = n \lambda_{\max}$. $(b)$ is due to $\norm{y}_2 = \norm{B_\Lambda x^{(o,d)}}_2 \leq \norm{B_\Lambda}_{op} \norm{x^{(o,d)}}_2 \leq n^2 \sqrt{m} \lambda_{\max}$. 
    
    Hence Assumption~\ref{assump:bounded_double_derivative} is satisfied with $C = 2 \lambda_{\max} \norm{Q}_{op} \sqrt{m} n (n+1) + \norm{c}_2$. 
\end{itemize}

\section{Differentially Private Network Routing Optimization}\label{sec:algorithm}

Given the setup from Section~\ref{sec:netflow_reformulate}, our objective is to design a request-level differentially private algorithm that returns a near optimal solution to \eqref{eqn:opt:random:omniscient}. Since the true distribution $\cQ$ of demand is unknown, we will design an algorithm for \eqref{eqn:opt:random:datadriven} and show that under the assumptions described in Section~\ref{sec:model:demand_model}, the algorithm's solution is also near optimal for \eqref{eqn:opt:random:omniscient}. 

Computing a near-optimal solution to \eqref{eqn:opt:random:datadriven} while maintaining differential privacy may seem like a daunting task, but it turns out that a single modification to a well-known optimization algorithm gives an accurate and differentially private solution. 

We present a Private Projected Stochastic Gradient Descent algorithm, which is described in Algorithm~\ref{alg:network_dpsgd}. As the name suggests, Algorithm~\ref{alg:network_dpsgd} is a modified version of stochastic gradient descent. The algorithm conducts a single pass over the historical data, using each data point to perform a noisy gradient step (see line 6). The key difference between Algorithm~\ref{alg:network_dpsgd} and standard stochastic gradient descent is in line 11, where instead of returning the final gradient descent iterate, Algorithm~\ref{alg:network_dpsgd} returns a noisy version of the last iterate. 
Algorithm~\ref{alg:network_dpsgd} has the following privacy and performance guarantees.

\begin{theorem}[Privacy Guarantee for Algorithm~\ref{alg:network_dpsgd}]\label{thm:dpsgd_privacy}
Algorithm~\ref{alg:network_dpsgd} is $(\epsilon,\delta)$-differentially private under request level adjacency defined in Definition~\ref{def:trip_level_adjacency}.
\end{theorem}

\begin{theorem}[Performance Guarantee for Algorithm~\ref{alg:network_dpsgd}]\label{thm:dpsgd_accuracy}
If $\Lambda_1,...,\Lambda_N \iid \cQ$, and $x^*$ is a solution to \eqref{eqn:opt:random:omniscient}, then the output $x_{\text{alg}}$ of Algorithm~\ref{alg:network_dpsgd} satisfies:
\begin{align*}
    \mathbb{E} \bigbra{ \norm{ x_{\text{alg}} - x^* }_2 } &\leq \frac{1}{\sqrt{N}} \frac{K C \exp \bigpar{ \frac{\beta^2 \pi^2}{12 \alpha^2}} }{\alpha} + \frac{n \sqrt{m}}{N} \bigpar{ \frac{\beta \exp \bigpar{ \frac{\beta^2 \pi^2}{12 \alpha^2} } }{\alpha \min(1,2\alpha)} + \frac{C}{\epsilon \alpha T} \sqrt{ 2\ln \bigpar{ \frac{1.25}{\delta} } } } \\
    &\leq O \bigpar{ \frac{1}{\sqrt{N}} } + O \bigpar{ \frac{1}{\epsilon N} \sqrt{\ln \frac{1}{\delta} } }
\end{align*}
\end{theorem}

\noindent In particular, $x_{\text{alg}}$ is $(\epsilon,\delta)$-differentially private and converges to $x^*$ as $N \rightarrow \infty$, meaning that privacy and asymptotic optimality are simultaneously achieved. 

See Appendix~\ref{pf:thm:dpsgd_privacy} and Appendix~\ref{pf:thm:dpsgd_accuracy} for proofs of Theorem~\ref{thm:dpsgd_privacy} and Theorem~\ref{thm:dpsgd_accuracy} respectively.

\begin{algorithm}
\caption{Private Projected Stochastic Gradient Descent}\label{alg:network_dpsgd}
\textbf{Input:} Historical demand data $\Lambda_1, ..., \Lambda_N$, Desired privacy level $(\epsilon,\delta)$\;
\textbf{Output:} Unit network flow $x \in \cX$\;
Initialize $x_0 \in \cX$ arbitrarily \;
\For{$1 \leq k \leq N$}{
    $\eta_{k-1} \leftarrow \min \bigpar{ \frac{1}{\alpha k}, \frac{\min(1,2\alpha)}{\beta} }$\;
    $x_k \leftarrow \Pi_\cX \bigpar{ x_{k-1} + \eta_{k-1} \nabla_x F(x_{k-1}, \Lambda_k)}$\;
}
$s \leftarrow \frac{C}{T} \min \bigpar{ \frac{\min(1,2\alpha)}{\beta}, \frac{1}{\alpha N} }$\;
$\sigma^2 \leftarrow 2 \frac{s^2}{\epsilon^2} \ln \bigpar{ \frac{1.25}{\delta} }$\;
$Z \sim \mathcal{N}\bigpar{0, \sigma^2 I}$ \;
$x_{\text{alg}} \leftarrow \Pi_\cX \bigpar{ x_N + Z} $\;
\textbf{Return} $x_{\text{alg}}$\;
\end{algorithm}

\subsection{Discussion}

Carefully adding noise to specific parts of existing algorithms is a principled approach for developing differentially private algorithms \cite{DworkMNS06, FeldmanMTT18, FeldmanKT20}. The main challenge in such an approach is determining a) where and b) how much noise to add. Suppose the goal is to (approximately) compute a query function $f(L)$ on a data set $L$ in a differentially private way. The latter question can be addressed by measuring the sensitivity of the desired query function. 

\begin{definition}[$\ell_2$ sensitivity]
Consider a function $f : \cL \rightarrow \mathbb{R}^d$ which maps data sets to real vectors. For a given adjacency relation $\text{Adj}$, the $\ell_2$ sensitivity of $f$, denoted $s_{\text{Adj}}(f)$, is the largest achievable difference in function value between adjacent data sets. Namely, 
\begin{align*}
    s_{\text{Adj}}(f) := \max_{ \substack{ L_1, L_2 \in \cL \\ \text{Adj}(L_1, L_2) = 1 } } \norm{f(L_1) - f(L_2)}_2.
\end{align*}
\end{definition}

Once the sensitivity of the query function is known, the required noise distribution can be determined using the Gaussian mechanism as described in Theorem \ref{thm:dwork_roth_14}.

\begin{theorem}[From \cite{DworkR14}]\label{thm:dwork_roth_14}
Suppose $f: \cD \rightarrow \mathbb{R}^p$ maps datasets to real vectors. If $s_{\text{Adj}}$ is the $\ell_2$ sensitivity of $f$, then $\widehat{f}(D) := f(D) + Z$ where $Z \sim \mathcal{N}\bigpar{0, 2 \frac{s_{\text{Adj}}^2}{\epsilon^2} \ln \bigpar{\frac{1.25}{\delta}} I_p}$ is $(\epsilon,\delta)$-differentially private with respect to the adjacency relation \text{Adj}. 
\end{theorem}

Calculating the sensitivity of the simple query functions (e.g., counting, voting, selecting the maximum value) is relatively easy, making the Gaussian mechanism straightforward to apply. However, for more complicated functions, noise calibration becomes more involved. 

Algorithm~\ref{alg:network_dpsgd} is an application of the Gaussian mechanism. Moreover, Theorem~\ref{thm:dpsgd_accuracy} shows that the suboptimality of Algorithm~\ref{alg:network_dpsgd} is $O \bigpar{\frac{1}{\sqrt{N}}}$. The asymptotic optimality of Algorithm~\ref{alg:network_dpsgd} comes from the fact that the $\ell_2$ sensitivity of the final gradient descent iterate is actually converging to zero as $N$ approaches infinity. This fact enables us to add less noise as $N \rightarrow \infty$. Indeed, the Gaussian noise added to the final gradient descent iterate in Algorithm~\ref{alg:network_dpsgd} has standard deviation which is $O \bigpar{ \frac{1}{N} }$.

In the remainder of this section, we sketch some of the mathematical ideas behind the perhaps non-intuitive result that the sensitivity of gradient descent converges to zero as the number of data points increases. For simplicity of exposition we will a) use scalar notation in place of vector notation for the sake of readability, and b) consider the simpler case of unconstrained gradient descent, which removes the need to perform projections. As a reminder, the full proof can be found in Appendix~\ref{pf:thm:dpsgd_privacy}.

Given a mobility data set, we use $x_t(L)$ to denote the $t$-th iterate of Algorithm~\ref{alg:network_dpsgd} when using data set $L$. It is sufficient to show that the $\max_t \abs{ \frac{d x_N}{d \Lambda_t} }$ converges to zero for every $t$. This is because the sensitivity is obtained by integrating the derivative:
\begin{align*}
    \norm{ x_N(L_1) - x_N(L_2) }_2 = \norm{ \int_{L_1}^{L_2} \frac{d x_N}{d L} dL }_2.
\end{align*}
Because $L_1, L_2$ are adjacent, the distance between them is finite, and thus the above integral will converge to zero if its integrand converges to zero. 

With this in mind, by chain rule, we can write
\begin{align*}
    \frac{d x_N}{d \Lambda_t} &= \frac{d x_t}{d \Lambda_t} \frac{d x_{t+1}}{d x_t} \frac{d x_{t+2}}{d x_{t+1}} ... \frac{d x_{N}}{d x_{N-1}}.
\end{align*}
Next, by using properties of smooth and strongly convex functions, we show that $\abs{ \frac{d x_{t+1}}{d x_t} } \leq 1 - \theta$ for some positive constant $\theta$. This result implies
\begin{align*}
    \abs{ \frac{d x_N}{d \Lambda_t} } &= \abs{\frac{d x_t}{d \Lambda_t}} (1-\theta)^{N-t}.
\end{align*}
Next, recalling that $x_t = x_{t-1} - \eta_t \nabla F(x_{t-1}, \Lambda_t)$, note that $x_{t-1}$ is computed from the first $t-1$ gradient steps, which only depend on the first $t-1$ data points $\Lambda_1, ..., \Lambda_{t-1}$. Hence the $\frac{d x_{t-1}}{d \Lambda_t} = 0$. From this we see that
\begin{align*}
    \frac{d x_t}{d \Lambda_t} &= \frac{d}{d \Lambda_t} \bigpar{ x_{t-1} - \eta_t \nabla F(x_{t-1}, \Lambda_t) } \\
    &= -\eta_t \frac{d}{d \Lambda_t} \nabla F(x_{t-1}, \Lambda_t).
\end{align*}
By using Assumption~\ref{assump:bounded_double_derivative} we have the following inequality:
\begin{align*}
    \abs{ \frac{d x_t}{d \Lambda_t} } = \abs{ \eta_t \frac{d}{d \Lambda_t} \nabla F(x_{t-1}, \Lambda_t) } \leq \eta_t C. 
\end{align*}
Putting everything together, we have
\begin{align*}
    \abs{ \frac{d x_N}{d \Lambda_t} } &= C \eta_t (1-\theta)^{N-t}.
\end{align*}
The choice of stepsize in Algorithm~\ref{alg:network_dpsgd} ensures three things: a) $\eta_t \leq 2$ for every $t$, b) $\eta_t$ is non-increasing, and c) $\eta_t \rightarrow 0$ as $t \rightarrow \infty$. Given these facts, there are two cases to consider for $\abs{ \frac{d x_N}{d \Lambda_t} }$.
\begin{enumerate}
    \item[Case 1:] $t \leq N/2$. In this case, we can upper bound $\eta_t \leq 2$ and $(1-\theta)^{N/2}$. Hence $\abs{ \frac{d x_N}{d \Lambda_t} } \leq 2 C (1-\theta)^{N/2}$. Since $1-\theta < 1$, this converges to zero as $N \rightarrow \infty$. 
    \item[Case 2:] $t \geq N/2$. In this case, we simply upper bound $(1-\theta)^{N-t} \leq 1$ and $\eta_t \leq \eta_{N/2}$ which gives $\abs{ \frac{d x_N}{d \Lambda_t} } \leq C \eta_{N/2}$ which converges to zero as $N \rightarrow \infty$. 
\end{enumerate}
Finally, this shows that
\begin{align*}
    \max_t \abs{ \frac{d x_N}{d \Lambda_t} } \leq \max \bigpar{ C \eta_{N/2}, 2 C (1-\theta)^N },  
\end{align*}
and since both terms in the maximum are converging to zero, we have $\lim_{N \rightarrow \infty} \abs{ \frac{d x_N}{d \Lambda_t} } = 0$. 

\section{Experiments}

Differentially private mechanisms add noise to provide a principled privacy guarantee for individual level user data. One immediate question is the degree to which the added noise degrades the quality of the obtained solution. In the previous section, we addressed this question theoretically in Theorem~\ref{thm:dpsgd_accuracy} by showing that Algorithm ~\ref{alg:network_dpsgd} is both differentially private and asymptotic optimal. In this section, we present empirical studies on privacy-performance trade-offs by comparing our algorithm's performance to that of a non-private network routing approach. To this end, we simulate a transportation network to evaluate the performance of our algorithm and the non-private optimal solution to the network flow problem \eqref{eqn:opt:random:datadriven}. 

We describe the dataset used for the experiments in Section \ref{sec:expt:data}. Next, the algorithms used in the experiments are described in Section \ref{sec:expt:alg}. In Section \ref{ssec:performance}, we evaluate the practical performance of our algorithm for different values of $N$. In particular, we study the effect of the number of samples, and quantify the loss in system performance we may experience due to the introduction of our privacy-preserving algorithm. Finally, in Section \ref{ssec:sensitivity}, we study the sensitivity of our algorithms to the simulation parameters. In particular, we study the convergence of the routing policy with increasing data for different demands, edge latency function, and the magnitude of the regularization loss introduced to convexify the edge latency function.

\subsection{Data Set}\label{sec:expt:data}

We use data for the Sioux Fall network, which is available in the Transportation Network Test Problem (TNTP) dataset \cite{tntp}. This network has 24 nodes, 76 edges, and 528 OD pairs (see Figure \ref{fig:sioux_falls} for an overview of the network topology). The distribution of mean hourly demand across different OD pairs is shown in Figure \ref{fig:sioux_falls}. The mean OD demand is 682 vehicles/ hour. Furthermore, for more context on the scale of the network, the travel time on edges ranges from 2 to 10 minutes. Trip data is generated each hour, i.e., $T = 60$, from a Poisson distribution with a mean value that is given by the data. Our objective is to learn a routing policy for these trips that minimizes the total travel costs for all users. To model congestion, we use the link latency model described in Section \ref{sec:model:transit-net-ex}. In particular, for each $e \in E$, we estimate the free flow latency as $c_e$ directly from the data set. The sensitivity of the latency function to the traffic volume on the link, denoted by $q_e$ is chosen such that the travel time on the link is doubled when the link flow equals the link capacity. In later experiments, we will change this factor to study the sensitivity of our algorithm.

\subsection{Algorithms}\label{sec:expt:alg}


Throughout Sections \ref{ssec:performance} and \ref{ssec:sensitivity} we compare the performance of two different algorithms: a Baseline algorithm and Algorithm 1, which we describe below. \\

\noindent \textbf{Baseline:} This algorithm computes the minimum travel time solution to \eqref{eqn:opt:random:datadriven} for the Sioux Fall model described in Section~\ref{sec:model:transit-net-ex}. Recall that in the Section \ref{sec:model:transit-net-ex} model, \eqref{eqn:transit-ex:time} describes the travel time incurred when serving demand $\Lambda$ with a routing policy $x$. Given a data set $\Lambda_1,...,\Lambda_N$, it computes the solution to the following minimization problem: $\min_{x \in \cX} \frac{1}{N} \sum_{i=1}^N x^\top B_{\Lambda_i}^\top Q B_{\Lambda_i} x + c^\top B_{\Lambda_i} x$. \\

\noindent \textbf{Algorithm 1:} The travel time function described in \eqref{eqn:transit-ex:time} is not strongly convex because $B_\Lambda^\top Q B_\Lambda$ is rank deficient. In order to satisfy Assumption~\ref{assump:strong_convex}, we introduce an $\ell_2$ regularization. Namely, given a data set $\Lambda_1,...,\Lambda_N$, we apply Algorithm~\ref{alg:network_dpsgd} to the following minimization problem: $\min_{x \in \cX} \alpha \norm{x}_2^2 + \frac{1}{N} \sum_{i=1}^N x^\top B_{\Lambda_i}^\top Q B_{\Lambda_i} x + c^\top B_{\Lambda_i} x$. \\

For the Sioux Falls network, the parameters for implementing Algorithm \ref{alg:network_dpsgd} are set as follows. The smoothness parameter $\beta$ is set to be the largest eigenvalue of $B_\Lambda^\top Q B_\Lambda$, which is equal to $2.08 \times 10^7$. We set $C=\beta$, and the regularizer parameter $\alpha = 10^{4}$. It is easy to check that these values satisfy the assumptions described in Section \ref{sec:model:transit-net-ex}. Note that several different values of $\alpha$ could have been used to convexify the latency function. However, our choice is governed by two factors. A small value of $\alpha$ will ensure that the regularized objective is a good estimate to the true objective, which is desirable. However, smaller values of $\alpha$ will lead to a larger condition number (which is $\beta/\alpha$), resulting in slower convergence and necessitating more data for achieving a similar performance. Thus, the particular value of $\alpha =10^{4}$ balances both these factors for our problem instance. In section \ref{ssec:sensitivity}, we present a sensitivity analysis with different values of $\alpha$.

\begin{figure}
     \centering
     \begin{subfigure}[c]{0.25\textwidth}
         \centering
         \includegraphics[trim={3cm, 3cm, 3cm, 3cm}, clip, width=\textwidth]{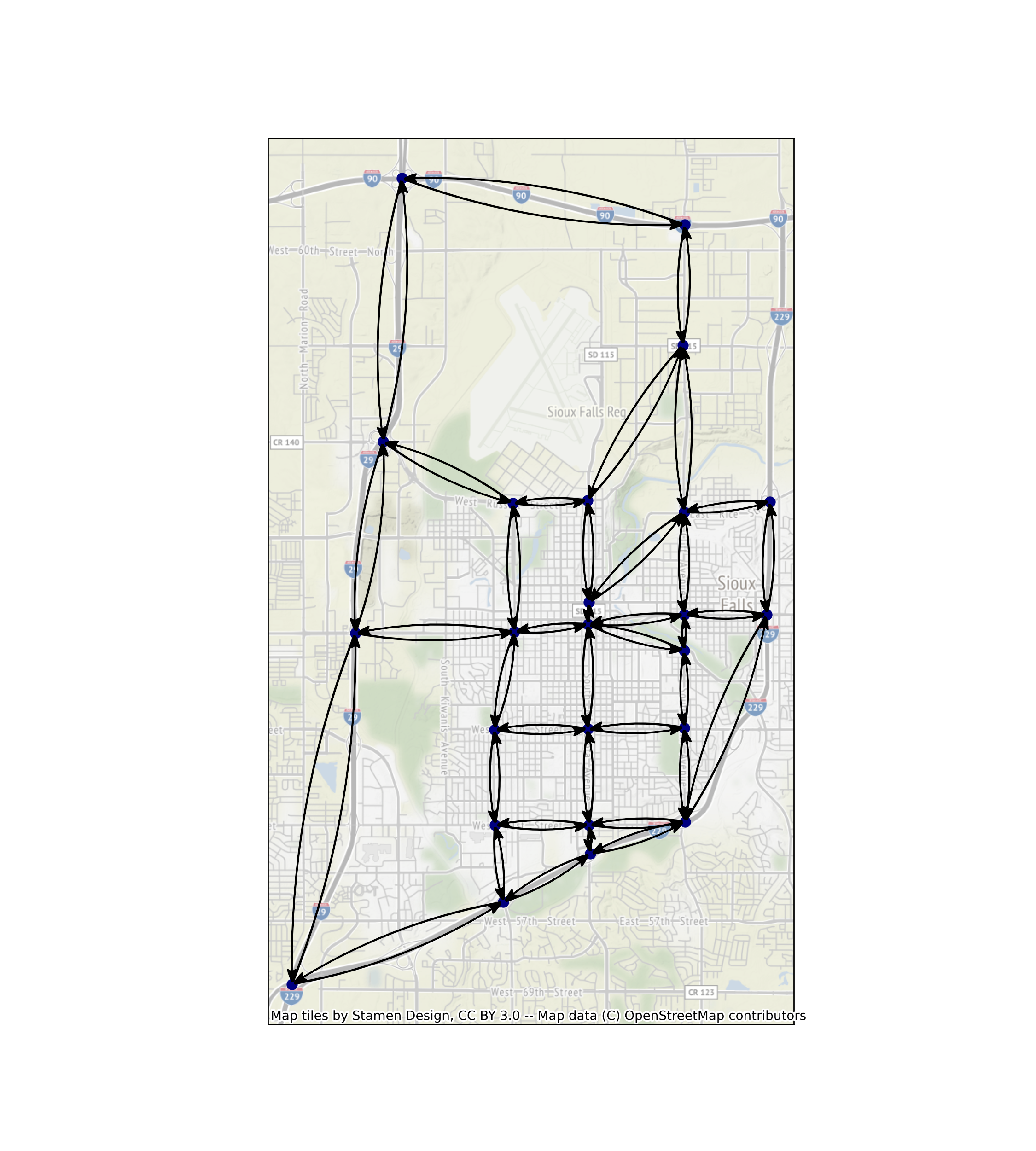}
         \caption{Network topology}
         \label{fig:topology}
     \end{subfigure}
     \hfill
     \begin{subfigure}[c]{0.36\textwidth}
         \centering
         \includegraphics[width=\textwidth]{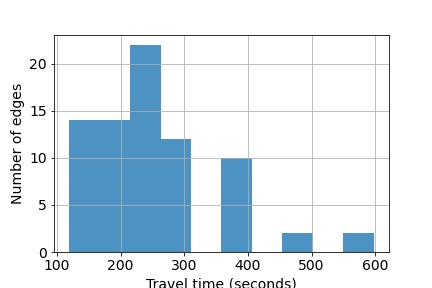}
         \caption{Travel time on edges}
         \label{fig:travel_time}
     \end{subfigure}
     \hfill
     \begin{subfigure}[c]{0.36\textwidth}
         \centering
         \includegraphics[width=\textwidth]{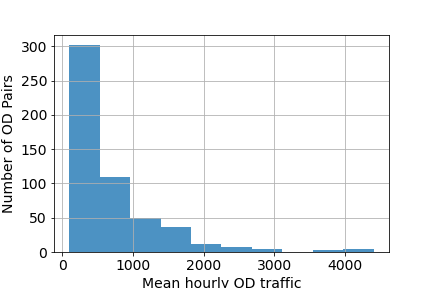}
         \caption{Mean demand between OD pairs}
         \label{fig:mean_demand}
     \end{subfigure}
        \caption{Description of the Sioux Falls dataset}
        \label{fig:sioux_falls}
\end{figure}

\subsection{Performance of Algorithm \ref{alg:network_dpsgd}} \label{ssec:performance}

In this subsection, we study the convergence of the routing policy with each step of the gradient descent performed by our algorithm. Since the impact of a routing policy is directly reflected in the total travel time, we plot the travel cost induced by a the learned routing policy as a function of the iterates. Recall that the number of iterations for a data set with $N$ data points is $N$ according to Algorithm \ref{alg:network_dpsgd}. In our experiments we evaluate the cost of a routing policy $x$ as $F(x, \frac{1}{N}\sum_{i=1}^{N} \Lambda_i)$ instead of using the sample average $\frac{1}{N}\sum_{i=1}^{N}F(x, \Lambda_i)$. This approximation is done solely for improving the run-time of our experiment (by up to 50X) and introduces less than $10^{-4}$\% error in the evaluation of the system costs. Similarly, the optimal solution is also computed by minimizing the objective function $F(x, \frac{1}{N}\sum_{i=1}^{N} \Lambda_i)$ instead of the term described in Equation \ref{eqn:opt:random:datadriven} to achieve a computational speed up of 30X. Evaluating the cost of this solution using the average-demand approximation results in errors less than $10^{-7}$\% error. Thus, in these experiments, we define the optimal costs as the approximate costs obtained through this procedure. Further details justifying these approximation are presented in Appendix \ref{app:approximations}

For our first set of experiments, we compare the objective values obtained by Algorithm \ref{alg:network_dpsgd} and the Baseline as a function of sample size $N$. In Figure \ref{fig:convergence_varyingN}, we plot the ratio of Algorithm \ref{alg:network_dpsgd}'s cost to Baseline's cost over the course of iterations for different values of $N$. For this set of experiments, we set the privacy parameters to $\epsilon=0.1$ and $\delta=0.1$. Note that Baseline's cost is fixed for a given $N$ and is computed offline to serve as a benchmark. For a given $N$, we only have $N$ iterations since each data point is only used once in Algorithm \ref{alg:network_dpsgd} to maintain privacy. For all three experiments ($N=10$, $N=25$, $N=50$), the cost decreases monotonically with additional iterations. It is therefore not surprising that the final costs for the $N=50$ case is the lowest, as we expect the routing policy learned with 50 data points to be better than the routing policy learned from 10 data points. It is however very interesting that even with a random routing policy initialization, our algorithm finds solutions that are just around 2\% away from the optimal policy. We suspect that this 2\% gap is due to the fact that the two algorithms have slightly different objective functions. Indeed, Algorithm~\ref{alg:network_dpsgd} has an $\ell_2$ regularizer in the objective, but Baseline does not. Our results therefore show that although the convergence is guaranteed only in the limit $N\rightarrow\infty$, we can obtain practically useful solutions with a relatively small number of data points.

\begin{figure}
    \centering
    \includegraphics[width=0.5\linewidth]{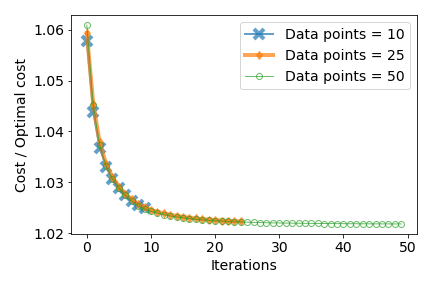}
    \caption{Convergence dynamics for different values of $N$}
    \label{fig:convergence_varyingN}
\end{figure}

In the next set of experiments, we study the effect of different privacy parameters on total travel time. To this end, we compare the costs of the pre-noise and post-noise solutions $x_N$ $x_{\text{alg}}$ from Algorithm~\ref{alg:network_dpsgd}. We conduct this comparison for $\epsilon \in \{0.01, 0.1, 0.5 \} $ and $\delta \in \{ 0.1, 0.5\}$. Table \ref{tab:privacy_costs} presents the percentage increase in total travel time due to the addition of privacy noise. 
The results in indicate that the price of privacy, i.e., the increase in total travel time due to the introduction of differential privacy noise is less than $7.8\times 10^{-2}$\% in the worst case. In fact, for more commonly used privacy parameters of $\epsilon=0.1$ and $\delta=0.1$, the cost of privacy is even smaller. One reason for this low cost of privacy is the high demand in the traffic network. From Figure \ref{fig:mean_demand}, we observe that every OD pair typically has a few hundred trips. With over 500 OD pairs, it is thus clear that there are tens of thousands of vehicles in the network contributing to trip information with every data point. Thus, with such a large number of vehicles, the noise required to protect the identity of one vehicle is not too high.

\begin{table}[]
    \centering
    \begin{tabular}{|c|c|c|}
    \hline
        $\epsilon$ &   $\delta$   & Cost (\% increase)\\ \hline
        0.01 & 0.1 & $7.83 \times 10^{-2}$ \\
        0.01 & 0.5 & $3.97 \times 10^{-3}$ \\
        0.1 & 0.1 & $9.06 \times 10^{-3}$ \\
        0.1 & 0.5 & $5.96 \times 10^{-3}$ \\
        0.5 & 0.1 & $2.44 \times 10^{-3}$ \\
        0.5 & 0.5 & $2.05 \times 10^{-3}$ \\ \hline
    \end{tabular}
    \caption{Change in routing costs due to incorporating privacy.}
    \label{tab:privacy_costs}
\end{table}

\subsection{Sensitivity}\label{ssec:sensitivity}

\begin{figure}
     \centering
     \begin{subfigure}[c]{0.32\textwidth}
         \centering
         \includegraphics[ width=\textwidth]{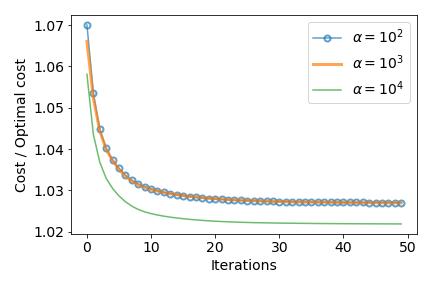}
         \caption{$\alpha$}
         \label{fig:sensitivity_alpha}
     \end{subfigure}
     \hfill
     \begin{subfigure}[c]{0.32\textwidth}
         \centering
         \includegraphics[width=\textwidth]{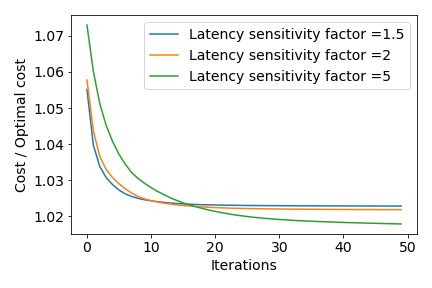}
         \caption{Slope of latency function}
         \label{fig:sensitivity_latency}
     \end{subfigure}
     \hfill
     \begin{subfigure}[c]{0.32\textwidth}
         \centering
         \includegraphics[width=\textwidth]{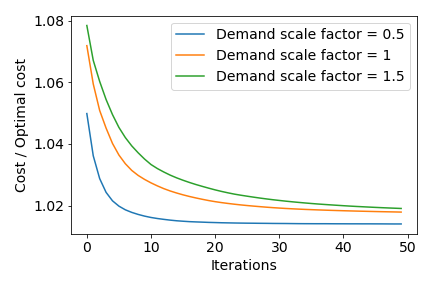}
         \caption{Traffic demand}
         \label{fig:sensitivity_demand}
     \end{subfigure}
        \caption{Sensitivity of the performance of Algorithm 1}
        \label{fig:sensitivity}
\end{figure}

We now study the sensitivity of our algorithm to input parameters. For these experiments, we fix the number of data points to $N=50$, since prior results suggest that most of the cost benefits are obtained with 50 data points. First, we study the effect of the regularizer term by setting $\alpha \in \{ 10^2, 10^3, 10^4\}$ and plot the convergence of the normalized costs in Figure \ref{fig:sensitivity_alpha}. We see that for larger values of $\alpha$, the costs decrease faster. This makes sense because larger $\alpha$ results in a lower condition number $\frac{\beta}{\alpha}$, which leads to faster convergence. We also note that the cost ratio does not go to 1 because Algorithm \ref{alg:network_dpsgd} is minimizing a regularized objective, while Baseline has no regularization.

In Figure \ref{fig:sensitivity_latency}, we compare three scenarios with varying slope for the latency function $q_e$. Recall that our previous experiments set the slope $q_e$ on each edge such that the travel time on the link doubles when the traffic is equal to the link capacity. This setting corresponds to a sensitivity factor of 2. We consider two more cases where where the travel time at capacity flow is 1.5 times and 5 times the free flow latency. Note that changing the latency sensitivity factor changes the matrix $Q$. Thus, for each of these experiments, we recompute the value of $\beta$ and $C$ and set it to be equal to the largest eigenvalue of the appropriate $B_\Lambda^\top Q B_\Lambda$. The optimal cost also varies for all the three cases and is recomputed. The value of $\alpha$ is fixed at $10^4$. We observe that when the latency function is steeper, i.e., the sensitivity factor is higher, the algorithm takes more iterations to reduce the costs, but eventually ends up with the lowest costs. This is because a larger $\beta$ leads to larger condition number $\frac{\beta}{\alpha}$, which makes convergence slow. However, a larger $\beta$ means that the $\ell_2$ regularizer is a smaller proportion of the total cost, meaning that the objectives of Algorithm~\ref{alg:network_dpsgd} and Baseline become more similar, which is what we believe causes the cost ratio to improve as the latency function becomes steeper.

Finally, we present the sensitivity of our algorithm to varying traffic demand in Figure \ref{fig:sensitivity_demand}. Elaborating further, in these experiments, we compare the nominal setting, where the demand is the mean demand with a scale factor of 1 to two cases. In the first case, we use a lower demand, where the mean traffic is 0.5 times the nominal traffic, and in the second case, the mean traffic is 1.5 times the nominal traffic. Again, as the demand changes, the matrix $B$ changes, and we recompute $\beta$ and $C$ as before. We observe that for the same value of $\alpha$, higher demand leads to better convergence and lower costs. This is because higher demand increases the travel time, making the $\ell_2$ regularizer a smaller proportion of Algorithm~\ref{alg:network_dpsgd}'s objective. The objective functions becoming more similar leads to the cost ratio being closer to $1$.


\section{Conclusions}
In this paper, we study the problem of learning network routing policies from sensitive user data. In particular, we consider the setting of a transportation network, where we want to learn and share a routing policy such that it does not reveal too much information about individual trips that may have contributed to learning this policy. Our paper presented a new approach to learn privacy-preserving routing policies by solving a reformulated network flow problem using a differentially private variant of the stochastic gradient descent algorithm. We prove that our algorithm is asymptotically optimal, meaning that the cost of the routing policy produced by our algorithm converges to the optimal non-private cost as the number of data points goes to infinity. Finally, our simulations on a Sioux Falls road network suggests that for realistic travel demands, we can learn differentially private routing policies that result in only a 2\% suboptimality in terms of total travel time. 

There are several interesting directions for future work. First, because differentially private algorithms are not allowed to be sensitive to single data points, they are naturally robust, and can be useful for tracking non-stationary demand distributions, as opposed to the stationary demand models we studied in this paper. In this paper, we studied request-level differential privacy where the goal is to occlude the influence of a single trip on the algorithm's output. Another practical and important notion is user-level differential privacy where the goal is to occlude the influence of all trips belonging to the same person on the algorithm's output. User-level privacy is harder to achieve, but is important in practice. Finally, generalizing the results to non-smooth objective functions would expand the domain of models that this technique can be applied to. 

\newpage 

\bibliographystyle{alpha}
\bibliography{peripherals/main}

\newpage 
\appendix
\section{Notation}\label{app:notation}

The terms and concepts defined and used throughout the paper are all described here for the sake of convenience. \\

\centerline{
\begin{tabular}{|c|p{0.85\textwidth}|}
    \hline
    Notation & Definition \\
    \hline
    $G$ & Graph representation of the network. (Definition~\ref{def:graph_representation}) \\
    \hline
    $V$ & The vertex set of the network $G$. (Definition~\ref{def:graph_representation}) \\
    \hline
    $E$ & The edge set of the network $G$. (Definition~\ref{def:graph_representation}) \\
    \hline 
    $n$ & The number of vertices in $G$, i.e., $n = \abs{V}$. (Definition~\ref{def:graph_representation}) \\
    \hline
    $m$ & The number of edges in $G$, i.e., $m = \abs{E}$. (Definition~\ref{def:graph_representation}) \\
    \hline
    $\cT$ & Operation period $\cT = \bigbra{t_{\text{start}}, t_{\text{end}}}$ which represents the time of day that the network operator is trying to optimize its decisions in. (Definition~\ref{def:operation-period}) \\
    \hline 
    $T$ & Number of minutes in the operation period. (Definition~\ref{def:operation-period}) \\
    \hline 
    $\Lambda$ & Demand matrix. $\Lambda(o,d)$ specifies the rate at which requests from $o \in V$ to $d \in V$ appear in the network. (Definition ~\ref{def:demand_matrix}) \\
    \hline
    $\cQ$ & The distribution of $\Lambda$. (Defined in Section~\ref{sec:model:demand_model}) \\
    \hline 
    $L$ & Dataset $L := (\Lambda_1,...,\Lambda_N)$ containing $N$ days of historical request data for the operation period $\cT$. (Defined in Section~\ref{sec:model:demand_model}) \\
    \hline 
    $x^{(o,d)}$ & Unit $(o,d)$ flow. A flow that routes exactly 1 unit of flow from $o$ to $d$ through the network. (Definition~\ref{def:unit_od_flow}) \\
    \hline 
    $x$ & Unit network flow. Specifies a unit $(o,d)$ flow for all pairs of vertices $(o,d)$ in the network. (Definition~\ref{def:unit_network_flow}) \\
    \hline 
    $\cX$ & The set of all unit network flows. It is also the feasible set for optimization problems~\eqref{eqn:opt:deterministic:omniscient},\eqref{eqn:opt:random:omniscient} and \eqref{eqn:opt:random:datadriven}. (Definition~\ref{def:unit_network_flow}) \\
    \hline     
    $x_k$ & A unit network flow and the $k$th gradient descent iterate from Algorithm~\ref{alg:network_dpsgd}. $x_k(L)$ is used to show explicit dependence of $x_k$ on the historical training data $L$. (Algorithm~\ref{alg:network_dpsgd}) \\
    \hline 
    $f_e$ & Link latency function for the edge $e \in E$. If $x$ is the total flow on $e$, then $f(x)$ will be the average travel time through the edge. (Definition~\ref{def:delay_functions}) \\
    \hline 
    $F$ & Total travel time function. $F(x,\Lambda)$ is the total travel time of requests in $\Lambda$ experience when they are routed according to $x$. (Defined in~\eqref{eqn:network_flow_to_time}) \\
    \hline 
    $(\epsilon, \delta)$ & Differential privacy parameters. As $\epsilon$ and $\delta$ get smaller, the privacy guarantees offered by differential privacy get stronger. (Definition~\ref{def:diffpriv}) \\
    \hline 
    $K$ & Upper bound on the standard deviation of the gradient of $F$. (Assumption~\ref{assump:bounded-var-grads}) \\
    \hline 
    $H$ & Hessian of $F$ with respect to $x$. Concretely, $H(x,\Lambda)$ is the hessian of $F$ with respect to $x$ evaluated at $(x,\Lambda)$. \\
    \hline 
    $\alpha$ & Strong convexity parameter. (Assumption~\ref{assump:strong_convex}) \\
    \hline 
    $\beta$ & Smoothness parameter. (Assumption~\ref{assump:smoothness} \\
    \hline 
    $C$ & Upper bound on $\norm{\cD_\Lambda \bigbra{\nabla F} (x,\Lambda)}_2$. (Assumption~\ref{assump:bounded_double_derivative}) \\
    \hline 
\end{tabular}
}

\newpage 
\section{Privacy Analysis (Proof of Theorem ~\ref{thm:dpsgd_privacy})}\label{pf:thm:dpsgd_privacy}


Recall the privacy guarantee for the Gaussian Mechanism from Theorem~\ref{thm:dwork_roth_14}. To show that Algorithm~\ref{alg:network_dpsgd} is $(\epsilon,\delta)$-differentially private, it suffices to show the following Lemma:

\begin{lemma}\label{lem:sgd_derivative_bound}
For every $1 \leq t \leq N$, any value of $ L := (\Lambda_1,...,\Lambda_N)$, and any $(o,d) \in V \times V$ we have $\norm{\cD_{\Lambda_t(o,d)} [x_N] (L) }_{op} \leq \min \bigpar{ \frac{C \min(1,2\alpha)}{\beta} , \frac{C}{\alpha N}}$.
\end{lemma}

Lemma~\ref{lem:sgd_derivative_bound} is sufficient to prove privacy, because it implies that the $\ell_2$ sensitivity of $x_N$ is at most $\frac{L}{\alpha N}$. To see why this is true, let $L_1 = (\Lambda_1,...,\Lambda'_t, ..., \Lambda_N)$ and $L_2 = (\Lambda_1,...,\Lambda''_t, ..., \Lambda_N)$ be request-level-adjacent data sets that differ only on the $t$th day. Furthermore, let $(o',d')$ be the request for which $\Lambda_t'$ and $\Lambda_t''$ differ. Let $x_N(L_1)$ and $x_N(L_2)$ be the final iterates of gradient descent obtained by using the data sets $L_1$ and $L_2$ respectively. By the fundamental theorem of calculus,
\begin{align*}
    x_N(L_2) - x_N(L_1) &= \int_{L_1}^{L_2} \cD_{\Lambda_t}[x_N](L) \; dL \\
    &= \int_{\Lambda'_t}^{\Lambda''_t} \cD_{\Lambda_t}[x_N](\Lambda_1,...,\Lambda_{t-1}, \Lambda, \Lambda_{t+1},...,\Lambda_N) \; d\Lambda \\
    \implies \norm{x_N^{(L')} - x_N^{(L)}}_2 &= \norm{ \int_{\Lambda'_t}^{\Lambda''_t} \cD_{\Lambda_t}[x_N](\Lambda_1,...,\Lambda_{t-1}, \Lambda, \Lambda_{t+1},...,\Lambda_N) \; d\Lambda }_2 \\
    &\leq \int_{\Lambda'_t}^{\Lambda''_t} \norm{ \cD_{\Lambda_t}[x_N](\Lambda_1,...,\Lambda_{t-1}, \Lambda, \Lambda_{t+1},...,\Lambda_N) \; d\Lambda }_2 \\
    &=\int_{\Lambda'_t(o',d')}^{\Lambda''_t(o',d')} \norm{ \cD_{\Lambda_t(o',d')}[x_N](\Lambda_1,...,\Lambda_{t-1}, \Lambda, \Lambda_{t+1},...,\Lambda_N) \; d\Lambda(o',d') }_2 \\
    &\leq \int_{\Lambda'_t(o',d')}^{\Lambda''_t(o',d')} \norm{ \cD_{\Lambda_t(o',d')}[x_N](\Lambda_1,...,\Lambda_{t-1}, \Lambda, \Lambda_{t+1},...,\Lambda_N) }_2 d\Lambda (o',d') \\
    &\overset{(a)}{\leq} \int_{\Lambda'_t}^{\Lambda''_t} \min \bigpar{ \frac{C \min(1,2\alpha)}{\beta} , \frac{C}{\alpha N}} d\Lambda (o',d') \\
    &\leq \min \bigpar{ \frac{C \min(1,2\alpha)}{\beta} , \frac{C}{\alpha N}} \abs{ \Lambda''_t(o',d') - \Lambda'_t(o',d') } \\
    &\overset{(b)}{\leq} \min \bigpar{ \frac{C \min(1,2\alpha)}{\beta T } , \frac{C}{\alpha T N}}.
\end{align*}
Here $(a)$ due to Lemma~\ref{lem:sgd_derivative_bound}. $(b)$ is because $\abs{ \Lambda''_t(o,d) - \Lambda'_t(o,d) } \leq T^{-1}$ is a consequence of $L_1,L_2$ being request-level-adjacent. 

Thus to establish $(\epsilon,\delta)$-differential privacy of Algorithm~\ref{alg:network_dpsgd}, all that remains is to prove Lemma~\ref{lem:sgd_derivative_bound}. 

\begin{proof}[Proof of Lemma~\ref{lem:sgd_derivative_bound}]

Define $L := (\Lambda_1,...,\Lambda_N)$. By Assumptions \ref{assump:twice_diff} , \ref{assump:strong_convex} and \ref{assump:smoothness}, the functions $F(\cdot, \Lambda_1), ..., F(\cdot, \Lambda_N)$ are all twice differentiable, $\alpha$-strongly convex and $\beta$-smooth. For each $t \leq N$, and any $(o,d) \in V \times V$ by chain rule we have

\begin{align*}
	\cD_{\Lambda_t(o,d)}[x_N](L) = \bigpar{ \prod_{k = t+1}^{N-1} \cD_{x_k}[x_{k+1}](L) } \cD_{\Lambda_t(o,d)}[x_{t+1}](L).
\end{align*}
Since $x_{k+1} = x_k - \eta_k \nabla_x F(x_k, \Lambda_k)$, differentiating both sides with respect to $x_k$ gives
\begin{align*}
	\cD_{x_k}[x_{k+1}](L) &= I - \eta_k H(x_k,\Lambda_k)
\end{align*}
where $H(\cdot, \Lambda_k)$ is the Hessian of $F(\cdot, \Lambda_k)$. Since $F(\cdot, \Lambda_k)$ is $\alpha$-strongly convex, we know $H(x_k, \Lambda_k) \succeq \alpha I$. By $\beta$-smoothness of $F(\cdot,\Lambda_k)$, we also know that $H(x_k, \Lambda_k) \preceq \beta I$. Therefore if $\eta_k \leq \frac{1}{\beta}$, we see that $\norm{\cD_{x_k}[x_{k+1}](L)}_{op} \leq 1-\eta_k \alpha$. From this we can conclude that

\begin{align}
\norm{\cD_{\Lambda_t(o,d)}[x_N](L)}_{op} &\leq \norm{\cD_{\Lambda_t(o,d)}[x_{t+1}](L)}_{op} \prod_{k = t+1}^{N-1} \norm{\cD_{x_k}[x_{k+1}](L)}_{op} \nonumber \\
&\leq \norm{\cD_{\Lambda_t(o,d)}[x_{t+1}](L)}_{op} \prod_{k = t+1}^{N-1} (1 - \eta_k \alpha) \nonumber \\
&\leq \norm{\cD_{\Lambda_t(o,d)}[x_{t+1}](L)}_{op} \prod_{k = t+1}^{N-1} \exp \bigpar{ - \eta_k \alpha} \nonumber \\
&= \norm{\cD_{\Lambda_t(o,d)}[x_{t+1}](L)}_{op} \exp \bigpar{ - \sum_{k = t+1}^{N-1} \eta_k \alpha} \label{eqn:chain_rule_contraction}
\end{align}

Finally, note that $x_k$ only depends on $x_0$ and $\Lambda_1,...,\Lambda_{k-1}$, and in particular it does not depend on $\Lambda_k$. Therefore differentiating both sides of $x_{k+1} = x_k - \eta_k \nabla_x F(x_k, \Lambda_k)$ with respect to $\Lambda_k(o,d)$ gives 
\begin{align}
    \cD_{\Lambda_t(o,d)}[x_{t+1}](L) &= -\eta_t \cD_{\Lambda_t(o,d)} \bigbra{ \nabla_x F} (x_t, \Lambda_t)  \nonumber \\
    \implies \norm{ \cD_{\Lambda_t(o,d)}[x_{t+1}] (L)}_{op} &= \eta_t \norm{ \cD_{\Lambda_t(o,d)} \bigbra{ \nabla_x F} (x_t, \Lambda_t)  }_{op} \overset{(a)}{\leq} C \eta_t, \label{eqn:initial_derivative}
\end{align}
where $(a)$ is due to Assumption~\ref{assump:bounded_double_derivative}. Combining inequalities \eqref{eqn:chain_rule_contraction} and \eqref{eqn:initial_derivative} gives
\begin{align}\label{eqn:final_iterate_derivative}
    \norm{\cD_{\Lambda_t(o,d)}[x_N] (L)}_{op} &\leq C \eta_t \exp \bigpar{ - \sum_{k = t+1}^{N-1} \eta_k \alpha}.
\end{align}

\noindent Letting $t_0 := \max \bigpar{ \frac{\beta}{\alpha \min(1,2\alpha)} - 1,0}$ note that
\begin{align*}
    \eta_t = \casewise{ 
        \begin{tabular}{cc}
            $\frac{1}{\alpha (t+1)}$ & if $t \geq t_0$ \\
            $\frac{\min(1,2\alpha)}{\beta}$ & otherwise. 
        \end{tabular}
    }
\end{align*}
If $t \geq t_0$, then $\eta_k = \frac{1}{\alpha(k+1)}$ for all $k \geq t$, and we see that
\begin{align*}
    \norm{\cD_{\Lambda_t(o,d)}[x_N] (L)}_{op} &\leq C \eta_t \exp \bigpar{ - \sum_{k = t+1}^{N-1} \eta_k \alpha} \\
    &= \frac{C}{\alpha (t+1)} \exp \bigpar{ - \alpha \sum_{k=t+1}^{N-1} \frac{1}{\alpha (k+1)} } \\
    &\leq \frac{C}{\alpha (t+1)} \exp \bigpar{ - \alpha \int_{t+1}^{N} \frac{1}{ \alpha y} dy } \\
    &= \frac{C}{\alpha (t+1)} \exp \bigpar{ - \ln N + \ln(t+1) } \\
    &= \frac{C}{\alpha (t+1)} \frac{t+1}{N} = \frac{C}{\alpha N}
\end{align*}
On the other hand, if $t \leq t_0$, then $\eta_t = \frac{\min(1,2\alpha)}{\beta}$ and we see that
\begin{align*}
    \norm{\cD_{\Lambda_t(o,d)}[x_N] (L)}_{op} &\leq C \eta_t \exp \bigpar{ - \sum_{k = t+1}^{N-1} \eta_k \alpha} \\
    &= \frac{C \min(1,2\alpha)}{\beta} \exp \bigpar{ - \alpha \sum_{k=t+1}^{N-1} \frac{1}{\alpha (k+1)} } \\
    &\leq \frac{C \min(1,2\alpha)}{\beta} \exp \bigpar{ - \alpha \sum_{k=\bigfloor{t_0}+1}^{N-1} \frac{1}{\alpha (k+1)} } \\
    &\leq \frac{C \min(1,2\alpha)}{\beta} \exp \bigpar{ - \alpha \int_{t_0+1}^{N} \frac{1}{\alpha y} dy} \\
    &\leq \frac{C \min(1,2\alpha)}{\beta} \frac{t_0 + 1}{N} \\
    &= \frac{C \min(1,2\alpha)}{\beta} \frac{\beta}{ \alpha \min(1,2\alpha) N} = \frac{C}{\alpha N}.
\end{align*}
Thus in either case we have $\norm{\cD_{\Lambda_t(o,d)}[x_N] (L)}_{op} \leq \frac{C}{\alpha N}$. \\

\noindent Finally, note that \eqref{eqn:final_iterate_derivative} implies $\norm{\cD_{\Lambda_t(o,d)}[x_N] (L)}_{op} \leq C \eta_t \leq C \frac{\min(1,2\alpha)}{\beta}$. Combining these two bounds gives the desired result:
\begin{align*}
    \norm{\cD_{\Lambda_t}[x_N] (L)}_{op} \leq \min \bigpar{ \frac{C \min(1,2\alpha)}{\beta}, \frac{C}{\alpha N} }. 
\end{align*}

\end{proof}

\newpage 
\section{Performance Analysis (Proof of Theorem~\ref{thm:dpsgd_accuracy})}\label{pf:thm:dpsgd_accuracy}

Let $x^*$ be a solution to \eqref{eqn:opt:random:omniscient}. Note that

\begin{align*}
	&\norm{x_{k+1} - x^*}_2^2 \\
	&= \norm{\prod_{\cX} \bigpar{x_k - \eta_k \nabla F (x_k, \Lambda_k)} - x^*}_2^2 \\
	&\overset{(a)}{=} \norm{\prod_{\cX} \bigpar{x_k - \eta_k \nabla F (x_k, \Lambda_k)} - \prod_{\cX} \bigpar{x^* - \eta_k \nabla \mathbb{E}_{\Lambda \sim \cQ} \bigbra{F(x^*, \Lambda)} }}_2^2 \\
	&\overset{(b)}{\leq} \norm{x_k - \eta_k \nabla F (x_k, \Lambda_k) - \bigpar{x^* - \eta_k \nabla \mathbb{E}_{\Lambda \sim \cQ} \bigbra{F(x^*,\Lambda)} }}_2^2 \\
	&= \norm{x_k - x^* - \eta_k \bigpar{\nabla F_k (x_k, \Lambda_k) - \nabla \mathbb{E}_{\Lambda \sim \cQ} \bigbra{F(x^*, \Lambda)} }}_2^2 \\
	&= \norm{x_k - x^*}_2^2 + \eta_k^2 \norm{\nabla F (x_k, \Lambda_k) - \nabla \mathbb{E}_{\Lambda\sim\cQ} F(x^*,\Lambda)}_2^2 - 2 \eta_k \bigpar{\nabla F (x_k, \Lambda_k) - \nabla \mathbb{E}_{\Lambda\sim\cQ} F(x^*,\Lambda)}^\top(x_k - x^*)
\end{align*}
Taking expectation of both sides conditioned on $x_k$, we see that

\begin{align*}
&\mathbb{E} \bigbra{ \evaluate{ \norm{x_{k+1} - x^*}_2^2 } x_k } \\
&\leq \mathbb{E} \bigbra{  \evaluate{\norm{x_k - x^*}_2^2 + \underbrace{\eta_k^2 \norm{\nabla F (x_k, \Lambda_k) - \nabla \mathbb{E}_{\Lambda\sim\cQ} F(x^*,\Lambda)}_2^2}_{\text{Term 2}} \underbrace{ - 2 \eta_k \bigpar{\nabla F (x_k, \Lambda_k) - \nabla \mathbb{E}_{\Lambda\sim\cQ} F(x^*,\Lambda)}^\top(x_k - x^*)}_{\text{Term 2}}} x_k } 
\end{align*}

To show that $x_{k+1}$ is closer to $x^*$ than $x_k$ is, we will provide bounds for both Term 2 and Term 3. 

\subsection{Bounding Term 2}

To upper bound Term 2, noting that $\mathbb{E}_{\Lambda_k}[\nabla f_k(x)] = \nabla f(x)$ for all $x \in \cX$, we have

\begin{align}
&\eta_k^2 \mathbb{E} \bigbra{ \evaluate{\norm{\nabla F (x_k, \Lambda_k) - \nabla \mathbb{E}_{\Lambda \sim \cQ} \bigbra{F(x^*, \Lambda)} }_2^2} x_k } \nonumber \\
&= \eta_k^2 \mathbb{E} \bigbra{ \evaluate{\norm{ \underbrace{\nabla F(x_k, \Lambda_k) - \nabla \mathbb{E}_{\Lambda \sim \cQ} \bigbra{F(x_k, \Lambda)}}_{A} + \underbrace{\nabla \mathbb{E}_{\Lambda \sim \cQ} \bigbra{F(x_k, \Lambda)} - \nabla \mathbb{E}_{\Lambda \sim \cQ} \bigbra{F(x^*, \Lambda)}}_{B}}_2^2} x_k } \nonumber \\
&= \eta_k^2 \mathbb{E} \bigbra{ \evaluate{ \norm{A}_2^2 + \norm{B}_{2}^2 + 2 A^\top B } x_k } \label{eqn:diffpriv_term2_bound_intermediate}
\end{align}
By Assumption~\ref{assump:bounded_double_derivative} and the dominated convergence theorem, we have $\nabla \mathbb{E}_{\Lambda \sim \cQ} \bigbra{ F(x_k, \Lambda) } = \mathbb{E}_{\Lambda \sim \cQ} \bigbra{ \nabla F(x_k, \Lambda) }$. Hence we see that 
\begin{align*}
    \mathbb{E} \bigbra{ \evaluate{ A } x_k } = \mathbb{E}_{\Lambda_k \sim \cQ} \bigbra{ \evaluate{ F(x_k, \Lambda_k) } x_k } - \mathbb{E}_{\Lambda \sim \cQ} \bigbra{ \evaluate{ F(x_k, \Lambda) } x_k } = 0.
\end{align*}
From this observation we have the three following remarks:
\begin{enumerate}
    \item Since $A$ is a zero mean random vector, $\mathbb{E} \bigbra{ \evaluate{ \norm{A}_2^2 } x_k}$ is the variance of $\nabla F(x_k, \Lambda)$ given $x_k$. By Assumption~\ref{assump:bounded-var-grads}, $\mathbb{E}_{\Lambda \sim \cQ} \bigbra{ \norm{\nabla F(x,\Lambda)}_2^2} \leq K^2$ for any $x$, which implies that $\mathbb{E} \bigbra{ \evaluate{ \norm{A}_2^2 } x_k} \leq K^2$. 
    \item By Assumption~\ref{assump:smoothness}, $F(\cdot, \Lambda)$ is $\beta$-smooth for every $\Lambda$. $\beta$-smoothness of $F(\cdot, \Lambda)$ implies that $\nabla F(\cdot,\Lambda)$ is $\beta$-lipschitz. Thus we have
    \begin{align*}
        \mathbb{E} \bigbra{ \evaluate{ \norm{B}_2^2 } x_k } &= \mathbb{E} \bigbra{ \evaluate{ \norm{ \mathbb{E}_{\Lambda \sim \cQ} \bigbra{ \nabla F(x_k, \Lambda) - \nabla F(x^*, \Lambda)} }_2^2 } x_k } \\
        &\overset{(a)}{\leq}  \mathbb{E} \bigbra{ \evaluate{ \mathbb{E}_{\Lambda \sim \cQ} \bigbra{ \norm{ \nabla F(x_k, \Lambda) - \nabla F(x^*, \Lambda)} }_2^2 } x_k } \\
        &\overset{(b)}{\leq} \mathbb{E} \bigbra{ \evaluate{ \mathbb{E}_{\Lambda \sim \cQ} \bigbra{ \beta^2 \norm{ x_k - x^* }}_2^2 } x_k } \\
        &= \beta^2 \norm{ x_k - x^* }_2^2.
    \end{align*}
    where $(a)$ is due to Jensen's inequality and $(b)$ is due to $\nabla F(\cdot, \Lambda)$ being $\beta$-Lipschitz. 
    \item Conditioned on $x_k$, $A$ is zero mean and $B$ is constant, meaning that $A^\top B$ is a zero mean random vector. Therefore $\mathbb{E} [ \evaluate{A^\top B} x_k] = 0$. 
\end{enumerate}

\noindent Applying these three remarks to the inequality \eqref{eqn:diffpriv_term2_bound_intermediate} we see that
\begin{align}
    \eta_k^2 \mathbb{E} \bigbra{ \evaluate{\norm{\nabla F (x_k, \Lambda_k) - \nabla \mathbb{E}_{\Lambda \sim \cQ} \bigbra{F(x^*, \Lambda)} }_2^2} x_k } &\leq \eta_k^2 \bigpar{ K^2 + \beta^2 \norm{x_k - x^*}_2^2 }. \label{eqn:diffpriv_term2_bound} 
\end{align}

\subsection{Bounding Term 3}

By linearity of expectation, Term 3 is equal to
\begin{align*}
&-2 \eta_k \mathbb{E} \bigbra{  \evaluate{ \bigpar{\nabla F (x_k, \Lambda_k) - \nabla \mathbb{E}_{\Lambda \sim \cQ} \bigbra{F(x^*,\Lambda)}}^\top(x_k - x^*)} x_k } \\
&= -2 \eta_k  \bigpar{\mathbb{E} \bigbra{  \evaluate{\nabla f (x_k, \Lambda_k)} x_k} - \nabla \mathbb{E}_{\Lambda \sim \cQ} \bigbra{F(x^*,\Lambda })}^\top(x_k - x^*) \\
&= -2 \eta_k  \bigpar{ \mathbb{E}_{\Lambda \sim \cQ} \bigbra{\nabla F(x_k,\Lambda) - \nabla F(x^*,\Lambda)}}^\top(x_k - x^*) .
\end{align*}
Define $f(x) := \mathbb{E}_{\Lambda \sim \cQ}[F(x,\Lambda)]$. We can re-write the above equation as
\begin{align*}
    -2 \eta_k \mathbb{E} \bigbra{  \evaluate{ \bigpar{\nabla F (x_k, \Lambda_k) - \nabla \mathbb{E}_{\Lambda \sim \cQ} \bigbra{F(x^*,\Lambda)}}^\top(x_k - x^*)} x_k } &\leq - 2\eta_k \bigpar{ \nabla f(x_k) - \nabla f(x^*) }^\top(x_k - x^*)
\end{align*}

Since $F(\cdot, \Lambda)$ is $\alpha$-strongly convex for every $\Lambda$, we can conclude that $f$ is also $\alpha$-strongly convex. To upper bound Term 3, we use $\alpha$-strong convexity of $f$ to conclude that
\begin{align*}
	f(x_k) &\geq f(x^*) + \nabla f(x^*)^\top (x_k - x^*) + \frac{\alpha}{2} \norm{x_k - x^*}_2^2 \\
	&\text{and} \\ 
	f(x^*) &\geq f(x_k) + \nabla f(x_k)^\top (x^* - x_k) + \frac{\alpha}{2} \norm{x_k - x^*}_2^2.
\end{align*}
Adding these inequalities together gives
\begin{align*}
f(x_k) + f(x^*) &\geq f(x_k) + f(x^*) - \bigpar{ \nabla f(x_k) - \nabla f(x^*) }^\top (x_k - x^*) + \alpha \norm{x_k - x^*}_2^2 \\
\implies -\alpha \norm{x_k - x^*}_2^2 &\geq - \bigpar{ \nabla f(x_k) - \nabla f(x^*) }^\top (x_k - x^*).
\end{align*}
This inequality implies the following bound on Term 3:
\begin{align}
-2 \eta_k \mathbb{E} \bigbra{  \evaluate{ \bigpar{\nabla F (x_k, \Lambda_k) - \nabla \mathbb{E}_{\Lambda \sim \cQ} \bigbra{F(x^*,\Lambda)}}^\top(x_k - x^*)} x_k } &\leq -2 \alpha \eta_k \norm{x_k - x^*}_2^2. \label{eqn:diffpriv_term3_bound}
\end{align}
\subsection{Putting everything together}
Applying the bounds \eqref{eqn:diffpriv_term2_bound} and \eqref{eqn:diffpriv_term3_bound} on Term 2 and Term 3 respectively, we see that

\begin{align}
\mathbb{E} \bigbra{ \evaluate{ \norm{x_{k+1} - x^*}_2^2 } x_k } &\leq \norm{x_k - x^*}_2^2 + \eta_k^2 K^2 + \beta^2 \eta_k^2 \norm{x_k - x^*}_2^2 - 2 \alpha \eta_k \norm{x_k - x^*}_2^2 \nonumber \\
&= \bigpar{ 1 - 2 \alpha \eta_k + \beta^2 \eta_k^2 } \norm{x_k - x^*}_2^2 + \eta_k^2 K^2. \label{eqn:affine_recursion}
\end{align}
By the tower property of expectation, we can write
\begin{align}
    \mathbb{E} \bigbra{ \norm{x_{N} - x^*}_2^2 } &= \mathbb{E} \bigbra{ ... \evaluate{ \mathbb{E} \bigbra{ \evaluate{ \mathbb{E} \bigbra{ \evaluate{ \norm{x_{N} - x^*}_2^2 } x_{N-1} } } x_{N-2} } ... } x_0 }. \label{eqn:sgd_error_tower_expect}
\end{align}
Combining the recursive relation from \eqref{eqn:affine_recursion} with \eqref{eqn:sgd_error_tower_expect} we see that
\begin{align*}
    \mathbb{E} \bigbra{ \norm{x_{N} - x^*}_2^2 } &\leq \norm{x_0 - x^*}_2^2 \bigpar{ \prod_{t=0}^{N-1} 1 - 2 \alpha \eta_t + \beta^2 \eta_t^2 } + K^2 \sum_{t=0}^{N-1} \eta_t^2 \bigpar{ \prod_{k=t+1}^{N-1} 1 - 2\alpha \eta_k + \beta^2 \eta_k^2 } \\
    &\leq \norm{x_0 - x^*}_2^2 \exp \bigpar{ \sum_{t=0}^{N-1} - 2 \alpha \eta_t + \beta^2 \eta_t^2 } + K^2 \sum_{t=0}^{N-1} \eta_t^2 \exp \bigpar{ \sum_{k=t+1}^{N-1} - 2\alpha \eta_k + \beta^2 \eta_k^2 }. 
\end{align*}
Since we chose $\eta_k := \min \bigpar{ \frac{1}{\alpha k}, \frac{\min(1,2\alpha)}{\beta} }$, we have $\eta_k \leq \frac{1}{\alpha k}$. This means that $\sum_{k=1}^\infty \eta_k^2 \leq \sum_{k=1}^\infty \frac{1}{\alpha^2 k^2} = \frac{\pi^2}{6 \alpha^2}$. Thus defining $C_{\alpha,\beta} := \exp \bigpar{ \frac{\beta^2 \pi^2}{6\alpha^2} }$, we have
\begin{align*}
    \mathbb{E} \bigbra{ \norm{x_{N} - x^*}_2^2 } &\leq C_{\alpha,\beta} \norm{x_0 - x^*}_2^2 \exp \bigpar{ \sum_{t=0}^{N-1} - 2 \alpha \eta_t } + K^2 C_{\alpha,\beta} \sum_{t=0}^{N-1} \eta_t^2 \exp \bigpar{ \sum_{k=t+1}^{N-1} - 2\alpha \eta_k } \\
    &\overset{(a)}{\leq} C_{\alpha,\beta} \norm{x_0 - x^*}_2^2 \exp \bigpar{ \sum_{t=0}^{N-1} - 2 \alpha \eta_t } + K^2 C_{\alpha,\beta} \sum_{t=0}^{N-1} \bigpar{ \frac{C}{\alpha N} }^2 \\
    &= C_{\alpha,\beta} \norm{x_0 - x^*}_2^2 \exp \bigpar{ \sum_{t=0}^{N-1} - 2 \alpha \eta_t } + \frac{ K^2 C^2 C_{\alpha,\beta}}{\alpha^2 N} 
\end{align*}
where $(a)$ is because in Appendix~\ref{pf:thm:dpsgd_privacy} we showed that $\eta_t \exp \bigpar{ \sum_{k=t+1}^{N-1} - \alpha \eta_k } \leq \frac{C}{\alpha N}$ for all $ 0 \leq t \leq N$. Next, let $t_0 := \max \bigpar{ \frac{\beta}{\alpha \min(1,2\alpha)} - 1, 0 }$ so that $\eta_t = \frac{1}{\alpha t}$ if $t \geq t_0$ and $\eta_t = \frac{\min(1,2\alpha)}{\beta}$ otherwise. We then have
\begin{align*}
    \mathbb{E} \bigbra{ \norm{x_{N} - x^*}_2^2 } &\leq C_{\alpha,\beta} \norm{x_0 - x^*}_2^2 \exp \bigpar{ \sum_{t=0}^{N-1} - 2 \alpha \eta_t } + \frac{ K^2 C^2 C_{\alpha,\beta}}{\alpha^2 N} \\
    &= C_{\alpha,\beta} \norm{x_0 - x^*}_2^2 \exp \bigpar{ \sum_{t=0}^{t_0} - 2 \alpha \eta_t } \exp \bigpar{ \sum_{t=t_0+1}^{N-1} - 2 \alpha \eta_t } + \frac{ K^2 C^2 C_{\alpha,\beta}}{\alpha^2 N} \\
    &\leq C_{\alpha,\beta} \norm{x_0 - x^*}_2^2 \exp \bigpar{ 2 \alpha \sum_{t=t_0+1}^{N-1} - \eta_t } + \frac{ K^2 C^2 C_{\alpha,\beta}}{\alpha^2 N} \\
    &\leq C_{\alpha,\beta} \norm{x_0 - x^*}_2^2 \exp \bigpar{ - 2 \alpha \int_{t_0+1}^{N} \frac{1}{\alpha y} dy } + \frac{ K^2 C^2 C_{\alpha,\beta}}{\alpha^2 N} \\
    &= C_{\alpha,\beta} \norm{x_0 - x^*}_2^2 \bigpar{ \frac{t_0+1}{N} }^2 + \frac{ K^2 C^2 C_{\alpha,\beta}}{\alpha^2 N} \\
    &= \frac{ C_{\alpha,\beta} \norm{x_0 - x^*}_2^2 \beta^2}{ \bigpar{ \alpha \min(1,2\alpha) }^2 N^2}  + \frac{ K^2 C^2 C_{\alpha,\beta}}{\alpha^2 N}.
\end{align*}
Finally, we have
\begin{align*}
    \mathbb{E} \bigbra{ \norm{x_{\text{alg}} - x^*}_2 } &= \mathbb{E} \bigbra{ \norm{ \Pi_\cX \bigpar{ x_{N} + Z} - x^*}_2 } \\
    &= \mathbb{E} \bigbra{ \norm{ \Pi_\cX \bigpar{ x_{N} + Z} - \Pi_\cX(x^*) }_2 } \\
    &\leq \mathbb{E} \bigbra{ \norm{ x_{N} + Z - x^*}_2 } \\
    &\leq \mathbb{E} \bigbra{ \norm{x_{N} - x^*}_2 } + \mathbb{E} \bigbra{ \norm{Z}_2 } \\
    &\overset{(a)}{\leq} \sqrt{\mathbb{E} \bigbra{ \norm{x_{N} - x^*}_2^2}} + \sqrt{\mathbb{E} \bigbra{ \norm{Z}_2^2}} \\
    &\leq \sqrt{ \frac{ C_{\alpha,\beta} \norm{x_0 - x^*}_2^2 \beta^2}{ \bigpar{ \alpha \min(1,2\alpha) }^2 N^2}  + \frac{ K^2 C^2 C_{\alpha,\beta}}{\alpha^2 N} } + \sqrt{ n^2 m \frac{2 C^2}{\epsilon^2 \alpha^2 T^2 N^2} \ln \bigpar{ \frac{1.25}{\delta} } } \\
    &\leq  \frac{ \sqrt{C_{\alpha,\beta}} \norm{x_0 - x^*}_2 \beta}{ \bigpar{ \alpha \min(1,2\alpha) } N}  + \frac{ K C \sqrt{C_{\alpha,\beta}}}{\alpha \sqrt{N}}  + \frac{C n \sqrt{m}}{\epsilon \alpha T N} \sqrt{ 2 \ln \bigpar{ \frac{1.25}{\delta} } } \\
    &\overset{(b)}{\leq} \frac{ \sqrt{C_{\alpha,\beta}} n^2 m \beta}{ \bigpar{ \alpha \min(1,2\alpha) } N}  + \frac{ K C \sqrt{C_{\alpha,\beta}}}{\alpha \sqrt{N}}  + \frac{C n \sqrt{m}}{\epsilon \alpha T N} \sqrt{ 2 \ln \bigpar{ \frac{1.25}{\delta} } }
\end{align*}
where $(a)$ is due to Jensen's inequality and $(b)$ is due to the fact that $\norm{x'-x''}_2 \leq n \sqrt{m}$ for any pair of unit network flows $x',x''$. 

\newpage 
\section{Computational approximations}\label{app:approximations}

In our experiments, we make several approximations for computational tractability. In this section, we provide empirical evidence that these approximations are reasonable and do not introduce significant errors. For ease of discussion, we define the following notations.
\begin{itemize}
    \item $x^{opt}_{\alpha}$: Solution obtained by solving the optimization problem (7) with regularizer $\alpha$
    \item $x_{\alpha}$: Solution obtained by solving the optimization problem with the average demand and a regularizer $\alpha$
    \item $F_{\alpha}(x, \Lambda_i)$: Evaluating the routing policy $x$ on demand $\Lambda_i$ with a regularizer $\alpha$
    \item $\langle F_{\alpha}(x, \Lambda_i) \rangle$: Evaluating the average cost of the routing policy $x$ on the set of demands $\{\Lambda_1 \hdots, \Lambda_N \}$ with a regularizer $\alpha$. More precisely, $\langle F_{\alpha}(x, \Lambda_i)\rangle = \frac{1}{N}\sum_{i=1}^{N}F_{\alpha}(x, \Lambda_i)$
    \item $ F_{\alpha}(x, \langle \Lambda_i \rangle) $: Evaluating the cost of the routing policy $x$ on the average demands $\langle \Lambda_i \rangle = \frac{1}{N} \sum_{i=1}^{N} \Lambda_i$ with a regularizer $\alpha$
\end{itemize}

\renewcommand{\arraystretch}{1.2}
\begin{table}[]
    \centering
    \begin{tabular}{|c|c|c|c|}
    \hline
    &   \textbf{Error \#1} & \textbf{Error \#2} & \textbf{Error \#3} \\ \cline{2-4}
    $N$   & $\frac{\langle F_{\alpha}(x^{opt}_{\alpha}, \Lambda_i) \rangle - F_{\alpha}(x^{opt}_{\alpha}, \langle \Lambda_i \rangle) }{\langle F_{\alpha}(x^{opt}_{\alpha}, \Lambda_i) \rangle}$   &  $\frac{\langle F_{\alpha}(x_{\alpha},  \Lambda_i ) \rangle - \langle F_{\alpha}(x^{opt}_{\alpha},  \Lambda_i ) \rangle }{\langle F_{\alpha}(x^{opt}_{\alpha},  \Lambda_i ) \rangle}$    & $\frac{\langle F_{\alpha=0}(x_{\alpha},  \Lambda_i )\rangle - \langle F_{\alpha=0}(x_{\alpha=0},  \Lambda_i ) \rangle }{\langle F_{\alpha=0}(x_{\alpha=0},  \Lambda_i )\rangle}$  \\ \hline
    10     &  $8.97\times10^{-7}$  & $<10^{-10}$  & $8.17\times10^{-3}$\\
    25     &  $1.02\times10^{-6}$  & $<10^{-10}$  & $8.16\times10^{-3}$\\
    50     &  $9.84\times10^{-7}$  & $<10^{-10}$  & $8.17\times10^{-3}$\\ \hline
    \end{tabular}
    \caption{Approximation errors.}
    \label{tab:approximation_errors}
\end{table}

Table \ref{tab:approximation_errors} presents errors from three different approximations. Our first approximation is to compute the system costs for a given policy by using the average demand instead of averaging the costs over every observed demand. The first column (denoted as Error \#) lists the fractional error introduced by this approximation for different values of $N$ when using the optimal routing policy. We note that the error is less than $10^{-6}$, and we we observe a 30-50X improvement in computational time when evaluating the costs using he average flow. This justifies the use of the average demand for estimating costs. 
Our second approximation is in solving an easier optimization problem to compute the optimal routing policy. In this case, the exact approach would be to solve the optimization problem described in Equation \ref{eqn:opt:random:datadriven}. However, the size of this problem grows rapidly with the number of data points $N$. Our approximation involves solving the easier optimization problem of maximizing $ F_{\alpha}(x, \langle \Lambda_i \rangle) $ to obtain the routing policy $x_{\alpha}$ instead of solving the original optimization problem to obtain $x_{\alpha}^{opt}$. The second column of the table (titled Error \#2) presents the error introduced due to this approximation on the travel costs. The small errors indicate that this assumption is reasonable, and helps us obtain upto a 30X speedup in solving the optimization problem.
Finally, we show through numerical evaluations that the addition of the $\alpha=10^3$ regularizer does not change the travel costs significantly (third column, denoted as Error \#3), and results in less than a 0.01\% error in the total travel cost.

\end{document}